\documentclass[10pt]{article}
\usepackage[margin=1in]{geometry}
\usepackage[utf8]{inputenc}
\usepackage[english]{babel}
\usepackage[dvipsnames]{xcolor}
\usepackage[normalem]{ulem}
\usepackage{microtype, graphicx, hyperref}
\usepackage[numbers]{natbib}
\usepackage[noline, boxed]{algorithm2e}
\usepackage{epstopdf}
\usepackage{bm, amssymb, amsmath, amsthm}
\newcommand{\I}{\mathcal{I}}   
\newcommand{\E}{\mathbb{E}}    
\newcommand{\bO}{\mathcal{O}}  
\newcommand{\Nd}{\mathcal{N}}  
\newcommand{\sv}{\, | \,}      
\newcommand{\eps}{\varepsilon} 
\newcommand{\bx}{\bm{x}}       
\newcommand{\by}{\bm{y}}       
\newcommand{\bv}{\bm{v}}       
\newcommand{\bu}{\bm{u}}       
\newcommand{\bz}{\bm{z}}       
\newcommand{\bth}{\bm{\theta}} 
\newcommand{\bS}{\bm{\Sigma}}  
\newcommand{\bR}{\bm{R}}       
\newcommand{\bOm}{\bm{\Omega}} 
\newcommand{\abs}[1]{\left|#1\right|}      
\newcommand{\br}[1]{\overline{#1}}         
\newcommand{\beps}{\bm{\eps}}              
\newcommand{\ind}[1]{\bm{1}_{\set{#1}}}    
\newcommand{\dif}{\mathop{}\!\mathrm{d}}   
\newcommand{\norm}[2][]{\left\Vert#2\right\Vert_{#1}} 
\newcommand{\set}[1]{\left\{ #1 \right\}}             
\newcommand{\mat}[1]{\begin{bmatrix}#1\end{bmatrix}}  

\newcommand{\tbOm}{\tilde{\bm{\Omega}}} 
\newcommand{\tL}{\tilde{\bm{L}}} 
\newcommand{\hbz}{\hat{\bz}}

\newcommand{\gee}{g_{\bth}}

\newcommand{\btau}{\bm{\tau}} 

\DeclareMathOperator*{\argmin}{arg\,min}   

\newtheorem{theorem}{Theorem}
\newtheorem*{theorem*}{Theorem}
\newtheorem{lemma}{Lemma}
\newtheorem*{lemma*}{Lemma}

\newtheorem*{corollary*}{Corollary}

\newtheorem*{definition*}{Definition}
\newtheorem{proposition}{Proposition}
\newtheorem*{proposition*}{Proposition}

\newtheorem*{conjecture*}{Conjecture}

\def\spacingset#1{\renewcommand{\baselinestretch}%
{#1}\small\normalsize} \spacingset{1}

\newcommand{\rev}[1]{\textcolor{black}{#1}}

\title{A Scalable Method to Exploit Screening in Gaussian Process Models with Noise}
\author{
Christopher J. Geoga
\quad
Michael L. Stein \\ \\
Dept. of Statistics, Rutgers University
}
\date{}

\begin{document}

\maketitle 

\begin{abstract}
A common approach to approximating Gaussian log-likelihoods at scale exploits
the fact that precision matrices can be well-approximated by sparse matrices in
some circumstances. This strategy is motivated by the \emph{screening effect},
which refers to the phenomenon in which the linear prediction of a process $Z$
at a point $\bx_0$ depends primarily on measurements nearest to $\bx_0$. But
simple perturbations, such as i.i.d.  measurement noise, can significantly
reduce the degree to which this exploitable phenomenon occurs.  While strategies
to cope with this issue already exist and are certainly improvements over
ignoring the problem, in this work we present a new one based on the EM
algorithm that offers several advantages. While in this work we focus on the
application to Vecchia's approximation \citep{vecchia1988}, a particularly
popular and powerful framework in which we can demonstrate true second-order
optimization of M steps, the method can also be applied using entirely
matrix-vector products, making it applicable to a very wide class of precision
matrix-based approximation methods. 
\end{abstract}

\bigskip


\newpage


\section{Introduction} \label{sec:intro}

Gaussian process models are increasingly important in a broad variety of fields.
In spatial statistics, they often arise via process models like
\begin{equation} \label{eq:model}
  Y(\bx) = Z(\bx) + \eps(\bx),
\end{equation}
where $Z$ is assumed to be a Gaussian process with nontrivial dependence
structure and $\eps$, typically independent of $Z$, is Gaussian white noise or
something with a comparably simple covariance structure. Throughout this work,
we will assume that $Z$ is mean-zero, although extensions to linear mean models
would be straightforward. When measured at a finite number of locations
$\set{\bx_j}_{j=1}^n$, say, this corresponds to the distributional model for
$\bz = \set{Z(\bz_j)}_{j=1}^n$ and $\beps$ defined similarly given by
$
  \by = \bz + \beps \sim \Nd\set{\bm{0}, \, \bS + \bR},
$
where $\bS$ is the covariance matrix for $\bz$ and $\bR$ denotes the covariance
matrix of $\beps$, with a prototypical example of $\bR = \eta^2 \I$, in which
case $\eps$ is often called a \emph{nugget} in the geostatistical literature. A
particularly common modeling strategy is to write a parametric model
$K_{\bth}(\cdot, \cdot)$ for the covariance function of $Z$, so that
$\bS_{j,k}(\bth) = \text{Cov}(Y(\bx_j), Y(\bx_k)) = K_{\bth}(\bx_j, \bx_k)$, and
then to numerically optimize the negative log-likelihood with respect to
parameters $\bth$, where the parameter set $\bth$ includes any parameters used
in $\bR$ as well.

When the data size $n$ is sufficiently large, operations with the typically
dense covariance matrix $\bS$ can be prohibitively expensive. Many direct matrix
approximations have been applied to this problem in the past in order to cope
with this issue, such as low-rank approximations \citep{cressie2008}, matrix
tapering \citep{kaufman2008} (one of the few approximating paradigms to have
supporting theory), implicit methods that employ fast algorithms like the fast
multipole method (FMM) and its descendants \citep{anitescu2012}, hierarchical
matrices \citep{ambikasaran2016,litvinenko2019,geoga2020,chen2021}, and many
others (to say nothing of dynamical approximations that are less directly
matrix-oriented \citep{wikle1999,stroud2001,katzfuss2012}). In several of these
strategies, adding a diagonal or otherwise sufficiently structured perturbation
is not problematic.  But direct covariance matrix-space methods come with
significant drawbacks, such as the difficulty of ensuring that approximations
for $\bS$ are positive definite while maintaining flexibility (see
\citep{chen2021} for a notable exception, however).  Not only is this necessary
in the sense of being required for a valid distributional model, it is
practically required because the Gaussian log-likelihood contains a
log-determinant. And while this issue can be avoided by solving the score
equations instead (see \cite{stein2013}) or using stochastic estimators for the
log-determinant based on Lanczos quadrature (see, \citep{ubaru2017}, or the
growing field of ``Bayesian optimization" \citep{mockus2012}), these estimators
can have difficulty being accurate enough for optimization.

The other primary thrust in approximating $\bS$ is to find sparse approximations
to $\bOm = \bS^{-1}$. Perhaps surprisingly, it is in some sense easier to obtain
valid positive-definite matrices with this approach than it is with the
hierarchical matrix approach, in part because such approximations are often
based on valid \emph{process} approximations. A good example of the
precision-based approximations is the class of so-called \emph{Vecchia
approximations} \citep{vecchia1988}, \rev{which will be the specific methods
that we will be using to showcase our estimation algorithm in this work}, which
effectively apply (compound-) Markovian-like assumptions that induce conditional
independence. Vecchia's approximation exploits the fact that any multivariate
density can be expanded in terms of conditional densities, in particular that
$
  p(y_1, y_2, ...,y_n) = p(y_1) \prod_{j=2}^n p(y_j \sv y_1, ..., y_{j-1}),
$
and attempts to approximate this quantity by instead conditioning on a subset of
points that hopefully are of similar predictive power as all of the past
observations, so that
$
  p(y_1, y_2, ...,y_n) \approx p(y_1) \prod_{j=2}^n p(y_j \sv y_{\sigma(j)}),
$
where $\sigma(j) \subseteq [j-1]$ and is typically $\bO(1)$ in size. With that
size constraint on the conditioning sets, evaluating the approximated
log-likelihood can be done in linear complexity (see
\citep{stein2004, katzfuss2021} for more details). While typically presented in
terms of writing many small likelihoods as above, Vecchia approximations can
also be considered through the lens of sparse approximations to precision
matrices \citep{sun2016, finley2019} or sparse symmetric factors of precision
matrices \citep{katzfuss2021, schafer2021}.  Good examples of methods that more
directly approximate the \emph{process}, as opposed to the approximating the
\emph{matrix} in some purely algebraic way, would be the so-called \emph{Markov
random field} (MRF) models \citep{rue2005}, which employ a graphical structure
to approximate the original process in a way that directly creates sparsity in
the precision matrix of the proxy process, and the closely related
SPDE/stochastic finite element approach \citep{lindgren2011, girolami2021}.

The issue of the perturbing noise $\eps$ in this latter paradigm is twofold.
Most substantively, sparse approximations to precision matrices---including the
Vecchia approximation introduced above---crucially depend on the \emph{screening
effect}, which is the phenomenon by which predictions depend very little on
far-away measurements when conditioned on nearby measurements
\citep{stein2002,stein2011}. Additive white noise, for example, severely reduces
the degree to which this phenomenon occurs \citep{stein2011,katzfuss2021}, and
so if such approximations are applied directly to the kernel with the added
diagonal perturbation, their accuracy with respect to typical assessments such
as the Kullback-Liebler (KL) divergence is significantly lowered
\citep{katzfuss2021}, although this admittedly does not necessarily imply that
the resulting estimators one obtains by maximizing the worse likelihood
approximation are in any sense ``worse" (see \citep{stein2004} for an example of
this phenomenon). In our experience, however, particularly with singleton
prediction sets (see \citep{stein2004} for discussion), point estimates are
indeed materially worsened in the sense of being farther from the MLE and having
a lower terminal log-likelihood.

The second issue is more practical: if the likelihood requires a log-determinant
and solving a linear system, it is difficult to get around the requirement for a
matrix factorization of some perturbation of $\tbOm$. As will be discussed
further in the next section, the typical perturbation is $\tbOm + \bR^{-1}$,
which needs to be factorized for the log-determinant at minimum. In the case of
Vecchia approximations, the current state-of-the-art provides methods for
directly assembling Cholesky (or, more generally, symmetric) factors of
(permutations of) $\tbOm$ so that $\tbOm = \tL \tL^T$, which means that when
there is not perturbative noise one can directly assemble the factorization and
avoid computing it with a sparse linear algebra library. But in the case of
perturbative noise, methods such as \citep{katzfuss2021} and \citep{schafer2021}
both unfortunately require first the assembly of $\tbOm = \tL \tL^T$ and then
the \emph{re}factorization of $(\tL \tL^T) + \bR^{-1}$ or some similar matrix,
bringing us back to the original problem. Sparse matrix factorizations have
several issues: for one, it is challenging to choose permutations of rows and
columns that leads to optimal sparsity of the factors \citep{saad2003}.  While
in practice the blow-up of non-zero elements is not always a concern, it can
also easily happen if one is not careful.  The second issue pertains primarily
to optimization, and addressing this problem is another focus of this work:
derivatives of symmetric factors of precision matrices with respect to kernel
parameters are challenging to work with. The simplest explanation for this is
that the derivative of Cholesky factor of a matrix requires several
matrix-matrix operations that are not easy to avoid \citep{murray2016}, which
even if one disregards the issue of ``fill-in," which refers to zero entries
becoming non-zero as the result of an algorithm, are computationally expensive,
particularly if one wishes to compute Hessian matrices of the likelihood.

In this work we present a new method for dealing with this issue that is
significantly broader and can be applied with no sparse matrix factorizations
inside of optimization routines, meaning their derivatives will not be required.
This is not only a practical gain in the sense of avoiding expensive operations
(and operations that can potentially harm complexity, as will be discussed later),
but it also enables effective automatic differentiation inside optimization
steps, making it possible to perform second-order optimization of likelihoods
whose derivatives are, at the least, very difficult to program efficiently by
hand. This is in contrast to the methods mentioned above, which either perform
parameter estimation using derivative-free methods or, more commonly, opt to use
Bayesian methods to perform parameter estimation. 

Before moving to a more specific comparison of our method with its most directly
related alternatives, we emphasize again that, \rev{while we apply this method
exclusively to Vecchia approximations in this work}, none of the tools are
actually specific to Vecchia approximations or to a specific structure of $\bR$
beyond something that admits fast solves and, ideally at least,
log-determinants. They have the potential to be useful in the SPDE framework
(see \cite{lindgren2011} for examples), the MRF framework \citep{rue2005},
``nearest-neighbor" Gaussian processes \citep{datta2016,finley2019}, and more,
where in at least some cases analytical methods to obtain symmetric factors are
not available, or where direct methods for solving linear systems with the
approximated matrices are best avoided.  

\subsection{Comparison with existing methods} \label{sec:comparison}

Since this paper will largely focus on the example setting of Vecchia
approximations, we first discuss the state-of-the-art for that specific
approximation. \rev{While they are very widely used, there are really only two
existing methodologies for attempting to address the way that measurement error
damages approximation accuracy for Vecchia approximations, at least for
likelihood-based methods.}

The most direct comparison would be to the ``sparse general
Vecchia" (SGV) method presented in \citep{katzfuss2021}. The idea of the SGV
approximation is to condition on a prudently selected combination of
noise-polluted measurements ($y_j = z_j + \eps_j$) and unobserved noise-free
measurements ($z_j$), and then to integrate out the unobserved values (note the
difference in our notation versus the notation in \citep{katzfuss2021},
however). The scheme that they propose for deciding when to condition on
observed values or unobserved values that need to be integrated out leads to
precision matrices with valuable theoretical guarantees on sparsity structure.
As an approximation of the probability distribution, it also performs
meaningfully better than a na\"ive application of Vecchia using the kernel with
the perturbation included in terms of Kullback-Liebler divergence
\citep{katzfuss2021}.

The other recent work on this topic is \citep{schafer2021}, which provides nice
theoretical guarantees about a specific subset of strategies for point ordering
and conditioning set selection being optimal with respect to KL divergence,
studying the \emph{maximin} ordering introduced in \citep{guinness2018}. The
work in \citep{schafer2021} was focused on studying properties of the matrix
approximation, but it does provide a suggested approach for dealing with the
issue of perturbative noise, which again involves assembling $\tbOm +
\bm{R}^{-1}$ and re-factorizing using an incomplete Cholesky factorization based
on the sparsity pattern of $\tbOm$.  This strategy seems sensible and some
heuristic justification was provided, but it was not the focus of the work and
few numerical details were provided on its accuracy. 

\rev{
In this work, we introduce and explore an alternative strategy that, at the cost
of a more complex estimation procedure (the EM algorithm), keeps the linear
complexity of standard Vecchia approximations while only applying the
approximation to $\bS$, and perhaps most importantly keeps the property that its
estimator is the solution to a set of unbiased estimating equations
\citep{stein2004,heyde2008}. While we will discuss its application only to
Vecchia approximations here, we again emphasize that this framework is equally
applicable to any method that seeks to supply a sparse approximation to
$\bS^{-1}$. Moreover, it is better seen as complementary to existing tools as
opposed to competitive due to the fact that it can simply be used to refine
other estimators.   
}

\section{Estimating covariance parameters using the EM algorithm} \label{sec:method}

\rev{The main idea of our method is to treat the problem of model \ref{eq:model}
as a missing data problem, where the ``missing" data is the data without the
perturbative noise $\beps$, and to use the EM algorithm to estimate parameters
\citep{dempster1977}. The primary observation we make in this section is to
provide a specific form of the ``E function" that is written in terms of
standard likelihoods with just $\bS$ or $\bR$ and an additional trace term that
is particularly well-suited to stochastic approximation}. Using the notation
introduced above, consider $Z$ to be the latent process, with a known covariance
function $K_{\bth}(\bx, \bx')$ but unknown parameters.  For generality, we will
use $\bR$ to denote the covariance matrix of finite-dimensional samples of the
noise $\eps$, as the methods discussed here apply in much more general settings
than $\bR$ diagonal.  We will assume in this work, however, that $\bR$ is full
rank.  For the duration of this section, we write this section completely
independently of any specific method for approximating $\bS$. In the interest of
readability, we will simply write $\bS^{-1}$ for the precision with the
understanding that applications will substitute the scalable $\tbOm \approx
\bS^{-1}$.  All of the linear algebra written here is exactly true regardless of
what $\bS$ looks like so long as the same matrices are used consistently in all
places. 

\rev{
The crux of the EM algorithm, at least in this setting, is preparing the ``E
function", which is the expectation of the joint log-likelihood of the observed
and missing data (in this case $\by$ and $\bz$) under the conditional law of
$\bz \sv \by$ using some fixed parameters $\bth_0$. By design, the EM algorithm
guarantees that any improvement in the E function gives at least that much
improvement to the marginal log-likelihood of $\by$ \citep{dempster1977}, and so
this naturally gives rise to an iterative process in which one prepares an E
function with parameters $\bth_0$, optimizes or improves said E function over
parameters $\bth$, and then repeats the process with the updated $\bth_0
\leftarrow \bth$.
}
While computing the expected joint log-likelihood of $\by$ and $\bz$ in the
interest of scalability may seem counter-intuitive, the following proposition,
whose elementary proof is provided in Appendix \ref{sec:proofs}, provides a
particularly convenient form for the E function.

\begin{proposition}
  The expected joint negative log-likelihood (``E function") of $(\by, \bz)$ at
  parameters $\bth$ under the law with $\bz \sv \by$ with covariance function
  $K_{\bth_0}$ is given by
  \begin{equation} \label{eq:joint_explike}
    \ell^{\E}_{\bth_0}(\bth)
    =
    \frac{1}{2}\mathrm{tr}\left[
      (\bS(\bth)^{-1} + \bR(\bth)^{-1})(\bS(\bth_0)^{-1} + \bR(\bth_0)^{-1})^{-1}
    \right]
    +
    \ell_{\bS(\bth)}(\hat{\bz}(\bth_0)) + \ell_{\bR(\bth)}(\by-\hat{\bz}(\bth_0)),
  \end{equation}
  where $\ell_{\bm{A}}(\bu)$ is the standard Gaussian negative log-likelihood for
  a vector $\bu$ and covariance matrix $\bm{A}$ and $\hbz(\bth_0) := (\bS(\bth_0)^{-1}
  + \bR(\bth_0)^{-1})^{-1} \bR(\bth_0)^{-1} \by$.
\end{proposition}

The full process for estimating parameters $\bth$ thus reduces to solving the
fixed-point problem $F(\bth) = \bth$, where $F(\bth_0) = \argmin_{\bth}
\ell^{\E}_{\bth_0}(\bth)$ (where evaluation of $F$ is sometimes called an ``M
step").  If one strictly maximizes the E function described here and performs
direct Picard iteration to solve the fixed point problem, this is exactly the
standard EM algorithm. If one simply improves the E function, perhaps by a fixed
number of Newton iterations, this procedure is a case of the generalized EM
algorithm \citep{dempster1977,neal1998}. \rev{The original theory from the EM
algorithm in \cite{dempster1977} gives that, in this setting, any local
minimizer of the marginal negative log-likelihood}
\begin{equation*} 
\rev{
  2\ell(\bth \sv \by) = \log \abs{\bS(\bth) + \bR(\bth)} + \by^T
  (\bS(\bth) + \bR(\bth))^{-1} \by
}
\end{equation*}
\rev{
is a fixed point of the above function $F$. The primary point of this additional
complexity is that this framework gives the practitioner the option to
approximate only $\bS(\bth)^{-1}$, as opposed to $(\bS(\bth) + \bR(\bth))^{-1}$,
and in doing so hopefully utilize much more accurate approximations and get
better estimates as a result. We provide the following further guarantee that
the estimates one obtains from this EM iteration corresponds to solving a set of
unbiased estimating equations if the approximation to $\bS(\bth)^{-1}$ gives
unbiased estimating equations in the noiseless case, which is trivially
satisfied by Vecchia approximations.
\begin{theorem}
  Let $\tbOm(\bth) \approx \bS(\bth)^{-1}$ be an approximated precision matrix
  with full rank that is differentiable with respect to $\bth \in \bm{\Theta}$
  and such that
  \begin{equation} \label{eq:appx_score}
    \E_{\bth} \nabla_{\bth} \set{
      - \log \abs{\tbOm(\bth)} + \bm{z}^T \tbOm(\bth) \bm{z}
    }
    = \bm{0},
  \end{equation}
  where $Z$ is a Gaussian process whose finite-dimensional covariance matrices
  are $\bS(\bth)$. Then if $Y = Z + \eps$ as in (\ref{eq:model}) and $\bR(\bth)$
  is indexed by non-overlapping parameters (so that either $\partial_{\theta_j}
  \tbOm(\bth) = \bm{0}$ or $\partial_{\theta_j} \bR(\bth) = \bm{0}$ for all
  $\bth$) and $\bm{\Theta}$ is open, the following two facts hold:
  \begin{enumerate}
  \item  Any local minimizer of the approximated negative log-likelihood
  \begin{equation} \label{eq:approx_nll}
      \log \abs{\tbOm(\bth)^{-1} + \bm{R}(\bth)} + \by^T (\tbOm(\bth)^{-1} +
      \bR(\bth))^{-1} \by
  \end{equation}
  is a fixed point of the EM iteration that also uses $\tbOm(\bth)$ in place of
  $\bS(\bth)^{-1}$.
  \item The gradient of (\ref{eq:approx_nll}) provides unbiased estimating
  equations for $\bth$.
  \end{enumerate}
\end{theorem}
The proof of this theorem is provided via a sequence of lemmas in Appendix
\ref{sec:proofs}.  If one were to know a priori that their approximated
log-likelihood were truly unimodal as well as some technical conditions on the
curvature of the log-likelihood surface detailed in Theorem $4$ of
\citep{dempster1977} in its subsequent discussion, conclusion $1$ of the theorem
could likely be strengthened to convergence to a global minimizer.
Unfortunately, even for simple models (like the Mat\'ern) and exact linear
algebra unimodality is not provable.  We have seen no evidence of multiple
minimizers for such models, however, and so there is good reason to expect that
the EM iterations will converge to the true minimizer of the approximated
log-likelihood.
}

\subsection{Symmetrized Hutchinson-type Trace Estimation} \label{sec:saa}

While it may seem problematic, the matrix that appears in the \rev{trace term of
the} E function, given by
\begin{equation} \label{eq:trace_matrix}
    (\bS(\bth)^{-1} + \bR(\bth)^{-1})(\bS(\bth_0)^{-1} + \bR(\bth_0)^{-1})^{-1},
\end{equation}
has several favorable properties that make stochastic trace estimation an
appealing option.  Hutchinson-type stochastic trace estimation
\citep{hutchinson1990} is based on the fact that $\E \bv^T \bm{A} \bv =
\mathrm{tr}(\bm{A})$ for any random vector $\bv$ such that $\E \bv = \bm{0}$ and
$\E \bv \bv^T = \I$. If one can draw samples from the law of $\bv$, then, the
empirical estimator $S^{-1} \sum_{j=1}^S \bv_j^T \bm{A} \bv_j$, often also
called the ``sample average approximation" (SAA) in the stochastic trace
estimation literature, will be an unbiased estimator for $\mathrm{tr}(\bm{A})$
\citep{hutchinson1990,avron2011,stein2013}.  The main challenge of these
estimators is that the variance of the estimator may depend on matrix properties
like (but not limited to) the condition number $\kappa(\bm{A})$ or its diagonal
concentration, so for a large and poorly conditioned matrix (like a large
covariance matrix) the performance can be unsuitable for something as precise as
maximum likelihood estimation.  A significant gain can be made within the
framework of Hutchinson-type estimators, however, by ``symmetrizing"; a common
fact about traces is that $\mathrm{tr}(\bm{A} \bm{B}) = \mathrm{tr}(\bm{B}
\bm{A})$, and so by extension if $\bm{A} = \bm{L} \bm{L}^T$ then
$\mathrm{tr}(\bm{A} \bm{B}) = \mathrm{tr}(\bm{L}^T \bm{B} \bm{L}) =
\mathrm{tr}(\bm{L} \bm{B} \bm{L}^T)$.  Hutchinson estimators for the Gaussian
score equations that use this ``symmetrized" form for the trace, for example,
have variances that are at worst equal to their un-symmetrized counterparts
\citep{stein2013}, but have been demonstrated to provide a significant gain in
practice \citep{geoga2020,geoga2021,beckman2022}. Further, in the notation of
the inline example, if $\bm{B}$ is positive definite, then $\bm{L} \bm{B}
\bm{L}^T$ is also positive definite, and in that case the error bounds of
Hutchinson-type estimators are improved further still \cite{roosta2015}.

The primary exploitable property of the matrix (\ref{eq:trace_matrix}) is that
there is good reason to believe that it is well-concentrated on the diagonal.
There are multiple popular options for the distribution of $\bv$, with
Rademacher($\tfrac{1}{2}$) (``random signs") and unit Gaussian entries being two
examples. While the variance of estimators is always smaller with
Rademacher($\tfrac{1}{2}$) entries, Gaussian-based estimators can have better
concentration properties \citep{avron2011}, and the best choice of sampling
vectors is matrix- or purpose-dependent. A particularly good use case for
Rademachaer entries, however, is when the matrix whose trace is being computed
is concentrated on the diagonal. For a general matrix $\bm{A}$, the variance of
$\bv^T \bm{A} \bv$ when $\bv$ has entries that are random signs is given by 
\begin{equation*} 
  \mathrm{Var}\left( \bv^T \bm{A} \bv \right)
  =
  2 \sum_{j \neq k}^n M_{j,k}^2
  =
  2 (\norm[F]{\bm{A}}^2 - \sum_{j=1}^n A_{j,j}^2),
\end{equation*}
compared with $2 \norm[F]{\bm{A}}^2$ when $\bv$ has Gaussian entries
\citep{hutchinson1990,avron2011}. While in general both of these variances can
be very large, it is clear that when $\bm{A}$ is nearly diagonal the variance
reduction can be substantial. And so for a matrix that looks like
$\bm{\Phi}(\bth) \bm{\Phi}(\bth_0)^{-1}$ that is smooth with respect to $\bth$
and with $\norm{\bth - \bth_0}$ small, it is reasonable to expect the variance
of this estimator, even before symmetrization, to be small. For this reason,
Rademacher($\tfrac{1}{2}$) entries perform very well in the setting of this
particular problem.

This particular setting is also one where pre-computation can be used to great
effect. \rev{Specifically, by observing that $\bth_0$ doesn't change for an entire M
step, it is possible and clearly beneficial to pre-compute
$\set{(\bS(\bth_0)^{-1} + \bR(\bth_0)^{-1})^{-1} \bv_j}_{j=1}^S$ and re-use
those pre-applied SAA vectors for the duration of the M step. This computation
can also be done in an inversion-free setting by using iterative methods,
further cutting down the necessary matrix operations to perform estimation.
}
Even better, if one is willing and able to compute a symmetric factorization
$\bS(\bth_0)^{-1} + \bR(\bth_0)^{-1} = \bm{W}(\bth_0) \bm{W}(\bth_0)^T$, where
$\bm{W}(\bth_0)$ is a symmetric factor that admits a fast solve, then one can
pre-compute the vectors $\tilde{\bv}_j = \bm{W}(\bth_0)^{-T} \bv_j$ to use
symmetrized trace estimators that come with potentially enhanced accuracy.
Unless an iterative method has been used to compute $\hbz(\bth_0)$, that
factorization likely has already been computed. \rev{The generic stochastic E
function in the symmetrized form can thus be computed as
\begin{equation} \label{eq:symmstoche}
  \ell^{\E}_{\bth_0}(\bth) \approx
  \frac{1}{2S} \sum_{j=1}^S \left[ \tilde{\bv}_j^T (\bS(\bth)^{-1} +
  \bR(\bth)^{-1})
  \tilde{\bv}_j \right] + \ell_{\bS(\bth)}(\hbz(\bth_0)) + \ell_{\br(\bth)}(\by
  - \hbz(\bth_0)),
\end{equation}
and the un-symmetric form can be computed as 
\begin{equation} \label{eq:asymstoche}
  \ell^{\E}_{\bth_0}(\bth) \approx
  \frac{1}{2S} \sum_{j=1}^S \left[ \bv_j^T (\bS(\bth)^{-1} + \bR(\bth)^{-1})
  \bar{\bv}_j \right] + \ell_{\bS(\bth)}(\hbz(\bth_0)) + \ell_{\br(\bth)}(\by
  - \hbz(\bth_0)),
\end{equation}
where $\bar{\bv}_j = (\bS(\bth_0)^{-1} + \bR(\bth_0)^{-1})^{-1} \bv_j$. In this
form, one can substitute a suitable approximation for $\bS(\bth)^{-1}$ and
obtain a stochastic E function that can be evaluated in the same complexity as
the relevant operations with the approximation to $\bS(\bth)^{-1}$ (assuming
that the relevant operations with $\bR$ are not a computational bottleneck).
Crucially, evaluating this expression requires no new matrix operations or
factorizations when compared to a standard log-likelihood approximation using a
sparse approximation $\tbOm \approx \bS^{-1}$.
}

\rev{
For the specific setting of Vecchia approximations, one can in fact do even
better and evaluate these quadratic forms without ever assembling a
$\bO(n)$-sized matrix to approximate $\bS(\bth)^{-1}$, instead using the
standard evaluation using sums of small conditional likelihoods that is trivial
to parallelize and automatically differentiate.  Specifically, let the standard
Vecchia approximation to $\ell_{\bS(\bth)}$ be given by
$
  2 \tilde{\ell}_{\bth}(\by) = 
    \tilde{\ell}^{\text{det}}(\bth)
    +
    \tilde{\ell}^{\text{qf}}_{\bth}(\by),
$
where the two terms correspond to the sum of small log-determinants and
quadratic forms respectively. Then simply by moving back and forth between the
standard form and the precision matrix form  of the Vecchia approximation (see
\cite{katzfuss2021} for full definitions), (\ref{eq:symmstoche}) can be further
improved to give
\begin{equation} \label{eq:enll_hutch}
  \ell^{\text{Vec}}_{\bth_0}(\bth) =
  \tilde{\ell}^{\text{det}}(\bth)
  +
  \tilde{\ell}^{\text{qf}}_{\bth}(\by)
  +
  \sum_{j=1}^S \tilde{\ell}^{\text{qf}}_{\bth}(S^{-1/2}\tilde{\bv}_j)
  +
  \frac{1}{2S}
  \sum_{j=1}^S \norm[2]{\sqrt{\bm{R}}(\bth)^{-T} \tilde{\bv}_j}^2
  +
  \ell_{\bR(\bth_0)}(\by - \hbz(\bth_0)).
\end{equation}
This specialized form of (\ref{eq:symmstoche}) gives a fully symmetrized and
unbiased estimator for $\ell^{\E}_{\bth_0}(\bth)$ that can be computed in a
single pass over each conditioning set and is matrix-free with respect to $\bS$.
}

\rev{
In summary, the cost of evaluating the stochastic estimator $ $ for the E
function $ $ is effectively equal to the cost of evaluating a standard Vecchia
likelihood for however many i.i.d. samples as one has SAA vectors. An important
point to consider when analyzing the additional computational burden that this
trace estimation poses is that the actual matrix operations of factorizing and
computing the quadratic forms is not necessarily the most costly part of
evaluating Vecchia-approximated likelihoods. Particularly for more complex
covariance functions, the kernel evaluations required to assemble each
conditional covariance are likely to dominate, and so in that sense the burden
of extra quadratic forms is unlikely to be problematic. As a basic example, the
difference in runtime cost between using $S=5$ and $S=150$ sampling vectors in
the problem setting of the next section (a standard four-parameter Mat\'ern
model) is a factor of $1.4$, not a factor of $30$.
}

\rev{
An important topic that requires discussion is the problem of picking the number
$M$ of sampling vectors $\set{\bv_j}_{j=1}^M$. Answering this question precisely
and in full generality is difficult if not impossible since it will depend on
specific aspects of the problem like the covariance function, dimension, and
sampling scheme. But with that said, there are always several easy empirical
diagnostics to assess whether or not the variability of the stochastic trace is
affecting estimates in a problematic way.  One option would be to simply add
more sampling vectors to your existing collection and perform an additional M
step with the same $\bth_0$ and compare the results. In the supplemental
material we offer a demonstration of this diagnostic for the problem setting of
the next section, with the results strongly indicating that even a very small
number of sampling vectors gives very stable estimates. As a rule of thumb, we
would ultimately suggest starting with approximately $70$, if for no other
reason than the fact that it perhaps provides a nice balance of variance
reduction and the diminishing returns of the $S^{-1/2}$ decay rate in standard
deviation.
}

\rev{
We close this section with a pseudo-code style summary of the algorithm to
estimate parameters $\bth$ given in Algorithm \ref{alg:em}, which for maximum
clarity is written in terms of $\tbOm \approx \bS^{-1}$ directly.
}
\begin{algorithm}
\textbf{Initialize:} 
Draw SAA vectors $\set{\bv_j}_{j=1}^M$, with $\bv_{j,k} \sim 
\text{Rademacher}(1/2)$  \\
\While{$\norm{\bth_{j+1} - \bth_j} > \epsilon$}{
  Compute conditional expectation $\hbz(\bth_j) := (\tbOm(\bth_j) +
  \bR(\bth_j)^{-1})^{-1} \bR(\bth_j)^{-1} \by$ 
  \\
  \eIf{Symmetrizing}{
    factorize $\tbOm(\bth_j) + \bR(\bth_j)^{-1} = \bm{W}(\bth_j)
    \bm{W}(\bth_j)^T$ \\
    Compute pre-(half-)solved SAA vectors $\tilde{\bv}_j = \bm{W}(\bth_j)^{-T} \bv_j$,
    \quad $j=1, ..., S$ \\
    Optimize or improve the symmetrized stochastic E function
    (\ref{eq:symmstoche}) to obtain $\bth_{j+1}$
  }{
    Compute pre-solved SAA vectors $\bar{\bv}_j = (\tbOm(\bth_j) +
    \bR(\bth_j)^{-1})^{-1} \bv_j$, \quad $j=1, ..., S$ \\
    Optimize or improve the non-symmetrized stochastic E function
    (\ref{eq:asymstoche}) to obtain $\bth_{j+1}$
  }
}
\caption{\rev{An overview of the procedure for estimating parameters $\bth$
using the EM algorithm and an approximate precision matrix $\tbOm(\bth) \approx
\bS(\bth)^{-1}$}.}
\label{alg:em}
\end{algorithm}

\section{Demonstrations} \label{sec:demo}

\subsection{A Comparison with SGV} \label{sec:comp_sgv}

In this section, we compare our estimation strategy to the one provided in
\citep{katzfuss2021} and made available in the \texttt{GPVecchia} R package.
The following experiment was run $50$ times:
\begin{enumerate}
\item Simulate $15{,}000$ points in two dimensions at random spatial locations in
$[0,1]^2$ of a Mat\'ern process, parameterized as $K_{\bth}(\bx, \bx') =
\sigma^2 \mathcal{M}_{\nu}(\rho^{-1} \norm{\bx - \bx'}) + \eta^2 \ind{\bx =
\bx'}$, where $\mathcal{M}_\nu$ is the Mat\'ern correlation function given by
\begin{equation*} 
  \mathcal{M}_\nu (t) = \frac{2^{1-\nu}}{\Gamma(\nu)} 
  (\sqrt{2 \nu} t)^\nu 
  \mathcal{K}_\nu 
  (\sqrt{2 \nu} t),
\end{equation*}
$\Gamma$ is the gamma function, and $\mathcal{K}_\nu$ is the modified
second-kind Bessel function of order $\nu$.  All simulations used fixed
parameters $\bth_{\text{true}} = (\sigma^2, \rho, \nu, \eta^2) = (10, 0.025,
2.25, 0.25)$.  \item Estimate the parameters using the SGV approximation via the
\texttt{vecchia\_estimate} function in \texttt{GPVecchia}, specified with the
maximin ordering \citep{guinness2018} and $10$ nearest-neighbor conditioning
points.  \item Estimate the parameters using our EM procedure\footnote{Code for
the EM procedure and a kernel-agnostic Vecchia likelihood, as well as all of the
scripts used to generate the results in this work, are available at
\texttt{https://github.com/cgeoga/Vecchia.jl}.} with the exact same Vecchia
approximation specification, initializing our EM iteration at the estimator
obtained by using the standard Vecchia approximation with the covariance
function that includes the nugget. For the EM procedure, allow up to $30$
iterations using $72$ SAA vectors with full symmetrization (although all
terminated in approximately $10$-$15$ iterations). For each M step, optimize
using true gradients and Hessians of $\ell^{\text{Vec}}_{\bth_0}(\bth)$
\rev{obtained via automatic differentiation \citep{Revels2016,geoga2022}.}
\end{enumerate}
\rev{
All computations were run on a computer with $16$ GiB of DDR4 RAM and an Intel
i5-11600K processor, using six threads (the number of physical cores). For
context, for data of this size and models of this parameter count, a single
gradient of the exact likelihood takes approximately two minutes to compute
using \texttt{ForwardDiff.jl} \citep{Revels2016}, even if one additionally
exploits manual derivative rules so that all matrix operations use the heavily
optimized LAPACK and OpenBLAS libraries. Even worse, evaluating the Hessian
requires more than 16 GiB of RAM if all second partial derivative matrices are
computed at once, and if only some were computed at a time the cost to
re-assemble derivative matrices as necessary would make performance even worse.
While a big enough computer with enough memory could of course solve this
estimation problem exactly, this setting is already well into the regime where
approximate methods are useful.
}

\begin{figure}[!ht]
  \centering
\begingroup
  \makeatletter
  \providecommand\color[2][]{%
    \GenericError{(gnuplot) \space\space\space\@spaces}{%
      Package color not loaded in conjunction with
      terminal option `colourtext'%
    }{See the gnuplot documentation for explanation.%
    }{Either use 'blacktext' in gnuplot or load the package
      color.sty in LaTeX.}%
    \renewcommand\color[2][]{}%
  }%
  \providecommand\includegraphics[2][]{%
    \GenericError{(gnuplot) \space\space\space\@spaces}{%
      Package graphicx or graphics not loaded%
    }{See the gnuplot documentation for explanation.%
    }{The gnuplot epslatex terminal needs graphicx.sty or graphics.sty.}%
    \renewcommand\includegraphics[2][]{}%
  }%
  \providecommand\rotatebox[2]{#2}%
  \@ifundefined{ifGPcolor}{%
    \newif\ifGPcolor
    \GPcolortrue
  }{}%
  \@ifundefined{ifGPblacktext}{%
    \newif\ifGPblacktext
    \GPblacktexttrue
  }{}%
  \let\gplgaddtomacro\g@addto@macro
  \gdef\gplbacktext{}%
  \gdef\gplfronttext{}%
  \makeatother
  \ifGPblacktext
    \def\colorrgb#1{}%
    \def\colorgray#1{}%
  \else
    \ifGPcolor
      \def\colorrgb#1{\color[rgb]{#1}}%
      \def\colorgray#1{\color[gray]{#1}}%
      \expandafter\def\csname LTw\endcsname{\color{white}}%
      \expandafter\def\csname LTb\endcsname{\color{black}}%
      \expandafter\def\csname LTa\endcsname{\color{black}}%
      \expandafter\def\csname LT0\endcsname{\color[rgb]{1,0,0}}%
      \expandafter\def\csname LT1\endcsname{\color[rgb]{0,1,0}}%
      \expandafter\def\csname LT2\endcsname{\color[rgb]{0,0,1}}%
      \expandafter\def\csname LT3\endcsname{\color[rgb]{1,0,1}}%
      \expandafter\def\csname LT4\endcsname{\color[rgb]{0,1,1}}%
      \expandafter\def\csname LT5\endcsname{\color[rgb]{1,1,0}}%
      \expandafter\def\csname LT6\endcsname{\color[rgb]{0,0,0}}%
      \expandafter\def\csname LT7\endcsname{\color[rgb]{1,0.3,0}}%
      \expandafter\def\csname LT8\endcsname{\color[rgb]{0.5,0.5,0.5}}%
    \else
      \def\colorrgb#1{\color{black}}%
      \def\colorgray#1{\color[gray]{#1}}%
      \expandafter\def\csname LTw\endcsname{\color{white}}%
      \expandafter\def\csname LTb\endcsname{\color{black}}%
      \expandafter\def\csname LTa\endcsname{\color{black}}%
      \expandafter\def\csname LT0\endcsname{\color{black}}%
      \expandafter\def\csname LT1\endcsname{\color{black}}%
      \expandafter\def\csname LT2\endcsname{\color{black}}%
      \expandafter\def\csname LT3\endcsname{\color{black}}%
      \expandafter\def\csname LT4\endcsname{\color{black}}%
      \expandafter\def\csname LT5\endcsname{\color{black}}%
      \expandafter\def\csname LT6\endcsname{\color{black}}%
      \expandafter\def\csname LT7\endcsname{\color{black}}%
      \expandafter\def\csname LT8\endcsname{\color{black}}%
    \fi
  \fi
    \setlength{\unitlength}{0.0500bp}%
    \ifx\gptboxheight\undefined%
      \newlength{\gptboxheight}%
      \newlength{\gptboxwidth}%
      \newsavebox{\gptboxtext}%
    \fi%
    \setlength{\fboxrule}{0.5pt}%
    \setlength{\fboxsep}{1pt}%
    \definecolor{tbcol}{rgb}{1,1,1}%
\begin{picture}(9636.00,6802.00)%
    \gplgaddtomacro\gplbacktext{%
      \csname LTb\endcsname
      \put(831,3741){\makebox(0,0)[r]{\strut{}\footnotesize -3}}%
      \put(831,4138){\makebox(0,0)[r]{\strut{}\footnotesize -2}}%
      \put(831,4534){\makebox(0,0)[r]{\strut{}\footnotesize -1}}%
      \put(831,4931){\makebox(0,0)[r]{\strut{}\footnotesize 0}}%
      \put(831,5327){\makebox(0,0)[r]{\strut{}\footnotesize 1}}%
      \put(831,5724){\makebox(0,0)[r]{\strut{}\footnotesize 2}}%
      \put(831,6120){\makebox(0,0)[r]{\strut{}\footnotesize 3}}%
      \put(963,3598){\makebox(0,0){\strut{}\footnotesize -3}}%
      \put(1541,3598){\makebox(0,0){\strut{}\footnotesize -2}}%
      \put(2119,3598){\makebox(0,0){\strut{}\footnotesize -1}}%
      \put(2698,3598){\makebox(0,0){\strut{}\footnotesize 0}}%
      \put(3276,3598){\makebox(0,0){\strut{}\footnotesize 1}}%
      \put(3854,3598){\makebox(0,0){\strut{}\footnotesize 2}}%
      \put(4432,3598){\makebox(0,0){\strut{}\footnotesize 3}}%
    }%
    \gplgaddtomacro\gplfronttext{%
      \csname LTb\endcsname
      \put(358,4930){\rotatebox{-270}{\makebox(0,0){\strut{}\footnotesize SGV}}}%
      \csname LTb\endcsname
      \put(2697,6285){\makebox(0,0){\strut{}\footnotesize scale $\hat{\sigma}^2 - \sigma_{\text{true}}^2$}}%
    }%
    \gplgaddtomacro\gplbacktext{%
      \csname LTb\endcsname
      \put(5071,3896){\makebox(0,0)[r]{\strut{}\footnotesize -0.02}}%
      \put(5071,4155){\makebox(0,0)[r]{\strut{}\footnotesize -0.015}}%
      \put(5071,4413){\makebox(0,0)[r]{\strut{}\footnotesize -0.01}}%
      \put(5071,4672){\makebox(0,0)[r]{\strut{}\footnotesize -0.005}}%
      \put(5071,4931){\makebox(0,0)[r]{\strut{}\footnotesize 0}}%
      \put(5071,5189){\makebox(0,0)[r]{\strut{}\footnotesize 0.005}}%
      \put(5071,5448){\makebox(0,0)[r]{\strut{}\footnotesize 0.01}}%
      \put(5071,5706){\makebox(0,0)[r]{\strut{}\footnotesize 0.015}}%
      \put(5071,5965){\makebox(0,0)[r]{\strut{}\footnotesize 0.02}}%
      \put(5429,3598){\makebox(0,0){\strut{}\footnotesize -0.02}}%
      \put(6183,3598){\makebox(0,0){\strut{}\footnotesize -0.01}}%
      \put(6937,3598){\makebox(0,0){\strut{}\footnotesize 0}}%
      \put(7691,3598){\makebox(0,0){\strut{}\footnotesize 0.01}}%
      \put(8445,3598){\makebox(0,0){\strut{}\footnotesize 0.02}}%
    }%
    \gplgaddtomacro\gplfronttext{%
      \csname LTb\endcsname
      \put(6937,6285){\makebox(0,0){\strut{}\footnotesize range $\hat{\rho} - \rho_{\text{true}}$}}%
    }%
    \gplgaddtomacro\gplbacktext{%
      \csname LTb\endcsname
      \put(831,812){\makebox(0,0)[r]{\strut{}\footnotesize -0.8}}%
      \put(831,1077){\makebox(0,0)[r]{\strut{}\footnotesize -0.6}}%
      \put(831,1341){\makebox(0,0)[r]{\strut{}\footnotesize -0.4}}%
      \put(831,1606){\makebox(0,0)[r]{\strut{}\footnotesize -0.2}}%
      \put(831,1870){\makebox(0,0)[r]{\strut{}\footnotesize 0}}%
      \put(831,2134){\makebox(0,0)[r]{\strut{}\footnotesize 0.2}}%
      \put(831,2399){\makebox(0,0)[r]{\strut{}\footnotesize 0.4}}%
      \put(831,2663){\makebox(0,0)[r]{\strut{}\footnotesize 0.6}}%
      \put(831,2928){\makebox(0,0)[r]{\strut{}\footnotesize 0.8}}%
      \put(1156,537){\makebox(0,0){\strut{}\footnotesize -0.8}}%
      \put(1541,537){\makebox(0,0){\strut{}\footnotesize -0.6}}%
      \put(1927,537){\makebox(0,0){\strut{}\footnotesize -0.4}}%
      \put(2312,537){\makebox(0,0){\strut{}\footnotesize -0.2}}%
      \put(2698,537){\makebox(0,0){\strut{}\footnotesize 0}}%
      \put(3083,537){\makebox(0,0){\strut{}\footnotesize 0.2}}%
      \put(3468,537){\makebox(0,0){\strut{}\footnotesize 0.4}}%
      \put(3854,537){\makebox(0,0){\strut{}\footnotesize 0.6}}%
      \put(4239,537){\makebox(0,0){\strut{}\footnotesize 0.8}}%
    }%
    \gplgaddtomacro\gplfronttext{%
      \csname LTb\endcsname
      \put(358,1870){\rotatebox{-270}{\makebox(0,0){\strut{}\footnotesize SGV}}}%
      \put(2697,130){\makebox(0,0){\strut{}\footnotesize EM}}%
      \csname LTb\endcsname
      \put(2697,3225){\makebox(0,0){\strut{}\footnotesize smoothness $\hat{\nu} - \nu_{\text{true}}$}}%
    }%
    \gplgaddtomacro\gplbacktext{%
      \csname LTb\endcsname
      \put(5071,718){\makebox(0,0)[r]{\strut{}\footnotesize -0.03}}%
      \put(5071,1102){\makebox(0,0)[r]{\strut{}\footnotesize -0.02}}%
      \put(5071,1486){\makebox(0,0)[r]{\strut{}\footnotesize -0.01}}%
      \put(5071,1870){\makebox(0,0)[r]{\strut{}\footnotesize 0}}%
      \put(5071,2254){\makebox(0,0)[r]{\strut{}\footnotesize 0.01}}%
      \put(5071,2638){\makebox(0,0)[r]{\strut{}\footnotesize 0.02}}%
      \put(5071,3022){\makebox(0,0)[r]{\strut{}\footnotesize 0.03}}%
      \put(5259,537){\makebox(0,0){\strut{}\footnotesize -0.03}}%
      \put(5818,537){\makebox(0,0){\strut{}\footnotesize -0.02}}%
      \put(6378,537){\makebox(0,0){\strut{}\footnotesize -0.01}}%
      \put(6937,537){\makebox(0,0){\strut{}\footnotesize 0}}%
      \put(7496,537){\makebox(0,0){\strut{}\footnotesize 0.01}}%
      \put(8056,537){\makebox(0,0){\strut{}\footnotesize 0.02}}%
      \put(8615,537){\makebox(0,0){\strut{}\footnotesize 0.03}}%
    }%
    \gplgaddtomacro\gplfronttext{%
      \csname LTb\endcsname
      \put(6937,130){\makebox(0,0){\strut{}\footnotesize EM}}%
      \csname LTb\endcsname
      \put(6937,3225){\makebox(0,0){\strut{}\footnotesize nugget $\hat{\eta}^2 - \eta_{\text{true}}^2$}}%
    }%
    \gplbacktext
    \put(0,0){\includegraphics[width={481.80bp},height={340.10bp}]{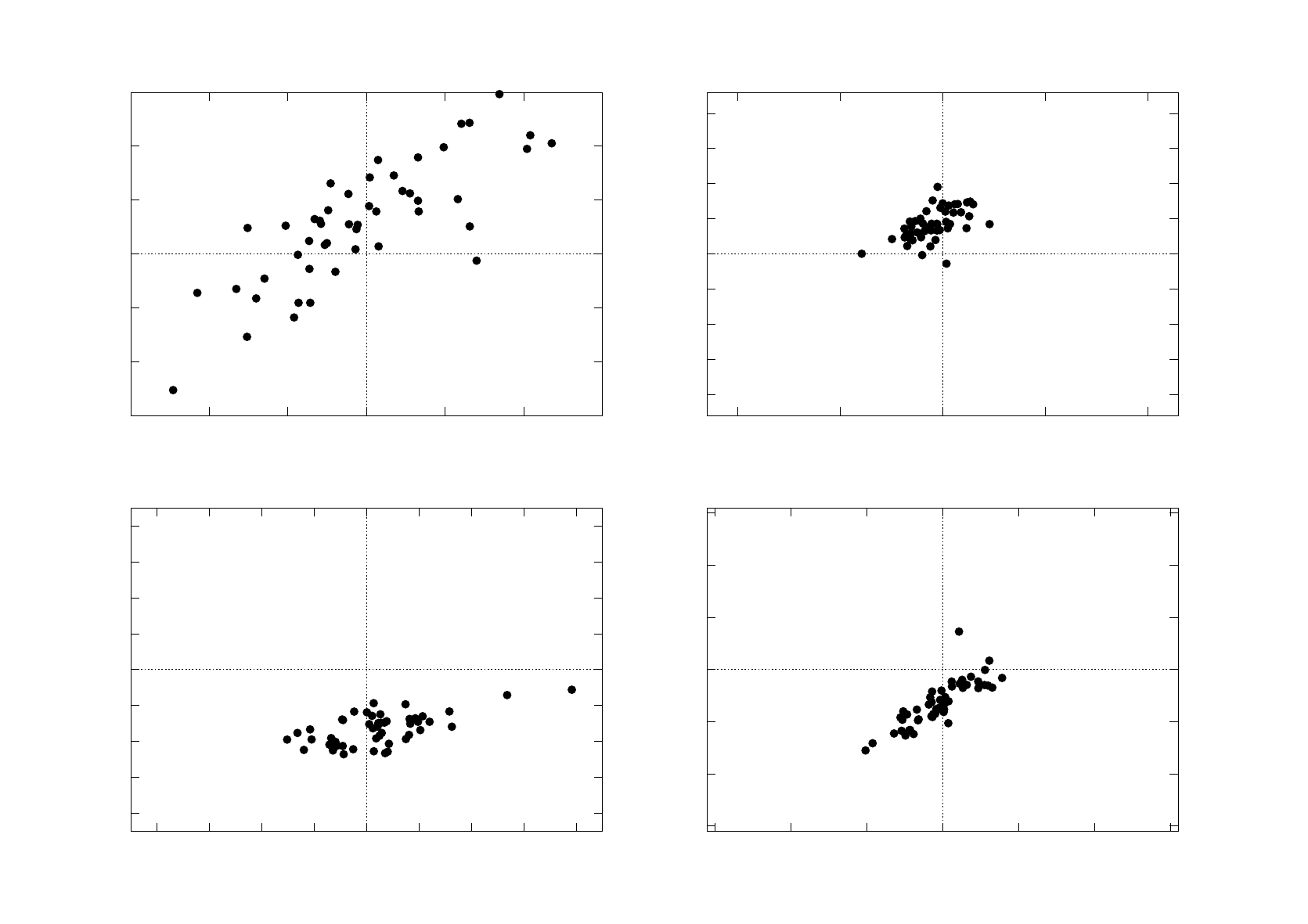}}%
    \gplfronttext
  \end{picture}%
\endgroup
  \caption{A summary of $50$ trials of the simulation study described above,
  where each point on the $(x,y)$ planes shows the difference of the EM-derived
  and SGV-derived point estimates respectively from the true parameter value.
  Clear bias is visible in both the range $\rho$ and smoothness $\nu$ for the
  SGV estimator, and is particularly pronounced in the latter.}
  \label{fig:sgv_compare}
\end{figure} 

\newpage

The results of the estimation study are summarized in Figure
\ref{fig:sgv_compare}, which shows scatter plots of pairs
$(\hat{\theta}_{\text{EM}} - \theta_{\text{true}}, \hat{\theta}_{\text{SGV}} -
\theta_{\text{true}})$, and several patterns immediately stand out.  The
estimates for the squared scale $\sigma^2$ agree reasonably well, but the range
parameter $\rho$ is at least slightly positively biased by SGV, a problem that
does not occur with the EM estimators. Similarly, the error variance $\eta^2$
seems negatively biased in the case of SGV. By far the biggest problem, however,
is that SGV-based estimates very consistently underestimate the smoothness
parameter $\nu$, to the point where every single estimate is below $2.25$. In
contrast, the EM-based estimator, while perhaps slightly more variable, does not
display any bias. While certain parameters are difficult to estimate and cannot
be estimated consistently under fixed domain asymptotics \citep{stein1999}, the
smoothness parameter $\nu$ in general can be estimated well in this setting, and
so considering how persistent this bias is, it seems reasonable to conclude that
this bias is an artifact of the SGV approximation.  As a final observation to
this point, re-trying the estimation in the SGV case even when initializing on
the exact true value of $\nu$ gives the same result, indicating that the implied
likelihood surface is clearly not flat in that region.

\begin{figure}[!ht]
  \centering
\begingroup
  \makeatletter
  \providecommand\color[2][]{%
    \GenericError{(gnuplot) \space\space\space\@spaces}{%
      Package color not loaded in conjunction with
      terminal option `colourtext'%
    }{See the gnuplot documentation for explanation.%
    }{Either use 'blacktext' in gnuplot or load the package
      color.sty in LaTeX.}%
    \renewcommand\color[2][]{}%
  }%
  \providecommand\includegraphics[2][]{%
    \GenericError{(gnuplot) \space\space\space\@spaces}{%
      Package graphicx or graphics not loaded%
    }{See the gnuplot documentation for explanation.%
    }{The gnuplot epslatex terminal needs graphicx.sty or graphics.sty.}%
    \renewcommand\includegraphics[2][]{}%
  }%
  \providecommand\rotatebox[2]{#2}%
  \@ifundefined{ifGPcolor}{%
    \newif\ifGPcolor
    \GPcolortrue
  }{}%
  \@ifundefined{ifGPblacktext}{%
    \newif\ifGPblacktext
    \GPblacktexttrue
  }{}%
  \let\gplgaddtomacro\g@addto@macro
  \gdef\gplbacktext{}%
  \gdef\gplfronttext{}%
  \makeatother
  \ifGPblacktext
    \def\colorrgb#1{}%
    \def\colorgray#1{}%
  \else
    \ifGPcolor
      \def\colorrgb#1{\color[rgb]{#1}}%
      \def\colorgray#1{\color[gray]{#1}}%
      \expandafter\def\csname LTw\endcsname{\color{white}}%
      \expandafter\def\csname LTb\endcsname{\color{black}}%
      \expandafter\def\csname LTa\endcsname{\color{black}}%
      \expandafter\def\csname LT0\endcsname{\color[rgb]{1,0,0}}%
      \expandafter\def\csname LT1\endcsname{\color[rgb]{0,1,0}}%
      \expandafter\def\csname LT2\endcsname{\color[rgb]{0,0,1}}%
      \expandafter\def\csname LT3\endcsname{\color[rgb]{1,0,1}}%
      \expandafter\def\csname LT4\endcsname{\color[rgb]{0,1,1}}%
      \expandafter\def\csname LT5\endcsname{\color[rgb]{1,1,0}}%
      \expandafter\def\csname LT6\endcsname{\color[rgb]{0,0,0}}%
      \expandafter\def\csname LT7\endcsname{\color[rgb]{1,0.3,0}}%
      \expandafter\def\csname LT8\endcsname{\color[rgb]{0.5,0.5,0.5}}%
    \else
      \def\colorrgb#1{\color{black}}%
      \def\colorgray#1{\color[gray]{#1}}%
      \expandafter\def\csname LTw\endcsname{\color{white}}%
      \expandafter\def\csname LTb\endcsname{\color{black}}%
      \expandafter\def\csname LTa\endcsname{\color{black}}%
      \expandafter\def\csname LT0\endcsname{\color{black}}%
      \expandafter\def\csname LT1\endcsname{\color{black}}%
      \expandafter\def\csname LT2\endcsname{\color{black}}%
      \expandafter\def\csname LT3\endcsname{\color{black}}%
      \expandafter\def\csname LT4\endcsname{\color{black}}%
      \expandafter\def\csname LT5\endcsname{\color{black}}%
      \expandafter\def\csname LT6\endcsname{\color{black}}%
      \expandafter\def\csname LT7\endcsname{\color{black}}%
      \expandafter\def\csname LT8\endcsname{\color{black}}%
    \fi
  \fi
    \setlength{\unitlength}{0.0500bp}%
    \ifx\gptboxheight\undefined%
      \newlength{\gptboxheight}%
      \newlength{\gptboxwidth}%
      \newsavebox{\gptboxtext}%
    \fi%
    \setlength{\fboxrule}{0.5pt}%
    \setlength{\fboxsep}{1pt}%
    \definecolor{tbcol}{rgb}{1,1,1}%
\begin{picture}(6802.00,3400.00)%
    \gplgaddtomacro\gplbacktext{%
      \csname LTb\endcsname
      \put(548,340){\makebox(0,0)[r]{\strut{}\footnotesize 0}}%
      \put(548,793){\makebox(0,0)[r]{\strut{}\footnotesize 2}}%
      \put(548,1246){\makebox(0,0)[r]{\strut{}\footnotesize 4}}%
      \put(548,1700){\makebox(0,0)[r]{\strut{}\footnotesize 6}}%
      \put(548,2153){\makebox(0,0)[r]{\strut{}\footnotesize 8}}%
      \put(548,2606){\makebox(0,0)[r]{\strut{}\footnotesize 10}}%
      \put(548,3059){\makebox(0,0)[r]{\strut{}\footnotesize 12}}%
      \put(680,197){\makebox(0,0){\strut{}\footnotesize -35}}%
      \put(1360,197){\makebox(0,0){\strut{}\footnotesize -30}}%
      \put(2040,197){\makebox(0,0){\strut{}\footnotesize -25}}%
      \put(2720,197){\makebox(0,0){\strut{}\footnotesize -20}}%
      \put(3400,197){\makebox(0,0){\strut{}\footnotesize -15}}%
      \put(4080,197){\makebox(0,0){\strut{}\footnotesize -10}}%
      \put(4760,197){\makebox(0,0){\strut{}\footnotesize -5}}%
      \put(5440,197){\makebox(0,0){\strut{}\footnotesize 0}}%
      \put(6120,197){\makebox(0,0){\strut{}\footnotesize 5}}%
    }%
    \gplgaddtomacro\gplfronttext{%
      \csname LTb\endcsname
      \put(75,1699){\rotatebox{-270}{\makebox(0,0){\strut{}\footnotesize Frequency (count)}}}%
      \put(3400,-100){\makebox(0,0){\strut{}\footnotesize $\ell_{\bS(\hat{\bth}_{\text{EM}}) + \bR(\hat{\bth}_{\text{EM}})}(\by) - \ell_{\bS(\hat{\bth}_{\text{SGV}}) + \bR(\hat{\bth}_{\text{SGV}})}(\by)$}}%
    }%
    \gplbacktext
    \put(0,0){\includegraphics[width={340.10bp},height={170.00bp}]{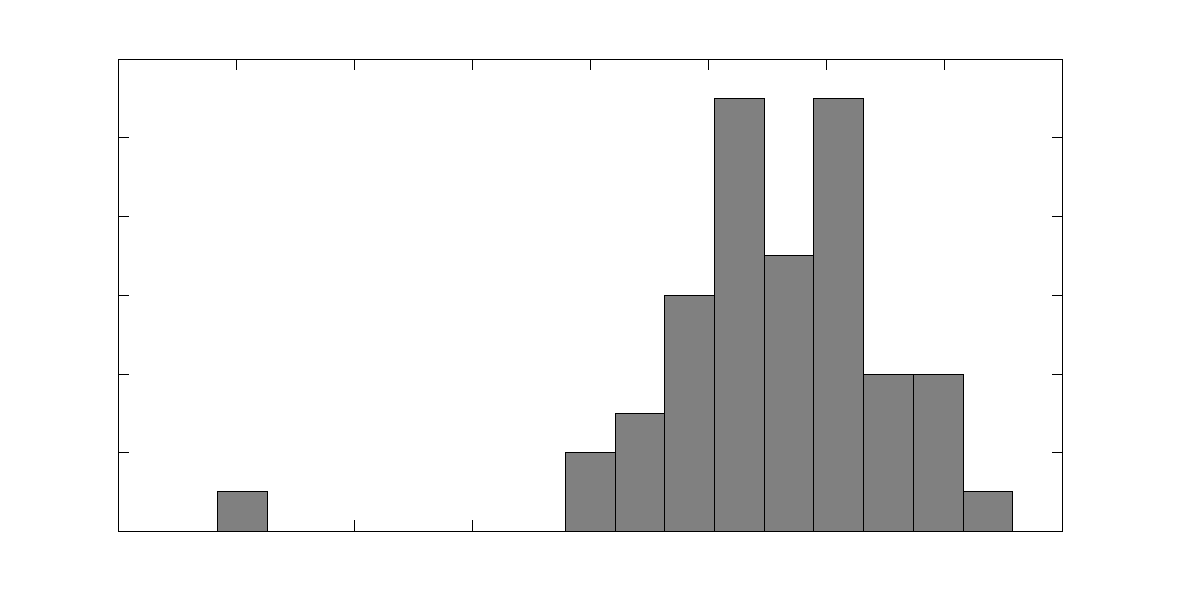}}%
    \gplfronttext
  \end{picture}%
\endgroup
  \caption{A histogram showing the difference in the exact negative
  log-likelihood between EM- and SGV-based parameters for the $50$ trials of the
  simulation study, with negative values indicating a superior terminal
  likelihood for the EM-based estimator.}
  \label{fig:sgv_compare_nll}
\end{figure}

Figure \ref{fig:sgv_compare_nll} shows a histogram of the difference in exact
negative log-likelihoods (assembled and computed with exact covariance matrices)
of the two estimates for each trial. In $49/50$ trials, the difference, computed
as
$
  \ell_{\bS(\hat{\bth}_{\text{EM}}) + \bR(\hat{\bth}_{\text{EM}})}(\by)
  -
  \ell_{\bS(\hat{\bth}_{\text{SGV}}) + \bR(\hat{\bth}_{\text{SGV}})}(\by),
$
is negative, indicating that the final likelihood of the EM-based estimators is
superior.  Considering the summary shown in Figure \ref{fig:sgv_compare}, this
is unsurprising, but the magnitude of these differences can be interpreted in
several different ways. On the one hand, for $15{,}000$ points, a difference of
$\approx 15$ log-likelihood units is very small, suggesting that the two methods
performed similarly. On the other hand, however, a difference of $15$ units has
the potential to be very meaningful, as the log-likelihood surface can be very
flat in sub-regions of the domain corresponding to parameters that have similar
interpolation properties but are very distant from the MLE in a pointwise sense
\citep{stein1999}. While that does not appear to have occurred here, a
difference of $15$ units absolutely merits serious attention, and the
consistently better likelihood values of the EM-based estimators is, in our
opinion, meaningful for this reason. If nothing else, any asymptotic theory
based on MLEs, regardless of whether or not it is applicable to MLEs computed in
this setting of spatial dependence, would absolutely not apply to an estimator
whose likelihood is sufficiently worse that it is not even approximately equal
to the MLE in finite computer precision.  \rev{With all of that said, however,
the EM-based estimators can take as much as twice as long to compute compared to
the SGV-based estimators, with runtimes ranging between thirty seconds and two
minutes compared to the consistent one minute of SGV}. 

As a final comparison, we consider the problem of predicting the process at the
center point $\bx_c = (1/2, 1/2)$ using $5{,}000$ nearest neighbors, with the
predictand denoted $z_c$. While not obviously the best method, it is standard
practice to perform prediction by estimating model parameters and then treating
them as the truth for subsequent operations like predicting. Figure
\ref{fig:centerinterp} provides a comparison of the absolute difference of the
true conditional mean from the plug-in prediction using estimated parameters for
each of the $50$ trials performed above, demonstrating the nontrivial effect
that the discrepancy in the estimated parameters can have on even the most
standard prediction problems.  In this comparison, the predictions made with the
EM-based estimators are almost uniformly better than the SGV-based estimates in
the sense of absolute difference from the true conditional mean.

\begin{figure}[!ht]
  \centering
\begingroup
  \makeatletter
  \providecommand\color[2][]{%
    \GenericError{(gnuplot) \space\space\space\@spaces}{%
      Package color not loaded in conjunction with
      terminal option `colourtext'%
    }{See the gnuplot documentation for explanation.%
    }{Either use 'blacktext' in gnuplot or load the package
      color.sty in LaTeX.}%
    \renewcommand\color[2][]{}%
  }%
  \providecommand\includegraphics[2][]{%
    \GenericError{(gnuplot) \space\space\space\@spaces}{%
      Package graphicx or graphics not loaded%
    }{See the gnuplot documentation for explanation.%
    }{The gnuplot epslatex terminal needs graphicx.sty or graphics.sty.}%
    \renewcommand\includegraphics[2][]{}%
  }%
  \providecommand\rotatebox[2]{#2}%
  \@ifundefined{ifGPcolor}{%
    \newif\ifGPcolor
    \GPcolortrue
  }{}%
  \@ifundefined{ifGPblacktext}{%
    \newif\ifGPblacktext
    \GPblacktexttrue
  }{}%
  \let\gplgaddtomacro\g@addto@macro
  \gdef\gplbacktext{}%
  \gdef\gplfronttext{}%
  \makeatother
  \ifGPblacktext
    \def\colorrgb#1{}%
    \def\colorgray#1{}%
  \else
    \ifGPcolor
      \def\colorrgb#1{\color[rgb]{#1}}%
      \def\colorgray#1{\color[gray]{#1}}%
      \expandafter\def\csname LTw\endcsname{\color{white}}%
      \expandafter\def\csname LTb\endcsname{\color{black}}%
      \expandafter\def\csname LTa\endcsname{\color{black}}%
      \expandafter\def\csname LT0\endcsname{\color[rgb]{1,0,0}}%
      \expandafter\def\csname LT1\endcsname{\color[rgb]{0,1,0}}%
      \expandafter\def\csname LT2\endcsname{\color[rgb]{0,0,1}}%
      \expandafter\def\csname LT3\endcsname{\color[rgb]{1,0,1}}%
      \expandafter\def\csname LT4\endcsname{\color[rgb]{0,1,1}}%
      \expandafter\def\csname LT5\endcsname{\color[rgb]{1,1,0}}%
      \expandafter\def\csname LT6\endcsname{\color[rgb]{0,0,0}}%
      \expandafter\def\csname LT7\endcsname{\color[rgb]{1,0.3,0}}%
      \expandafter\def\csname LT8\endcsname{\color[rgb]{0.5,0.5,0.5}}%
    \else
      \def\colorrgb#1{\color{black}}%
      \def\colorgray#1{\color[gray]{#1}}%
      \expandafter\def\csname LTw\endcsname{\color{white}}%
      \expandafter\def\csname LTb\endcsname{\color{black}}%
      \expandafter\def\csname LTa\endcsname{\color{black}}%
      \expandafter\def\csname LT0\endcsname{\color{black}}%
      \expandafter\def\csname LT1\endcsname{\color{black}}%
      \expandafter\def\csname LT2\endcsname{\color{black}}%
      \expandafter\def\csname LT3\endcsname{\color{black}}%
      \expandafter\def\csname LT4\endcsname{\color{black}}%
      \expandafter\def\csname LT5\endcsname{\color{black}}%
      \expandafter\def\csname LT6\endcsname{\color{black}}%
      \expandafter\def\csname LT7\endcsname{\color{black}}%
      \expandafter\def\csname LT8\endcsname{\color{black}}%
    \fi
  \fi
    \setlength{\unitlength}{0.0500bp}%
    \ifx\gptboxheight\undefined%
      \newlength{\gptboxheight}%
      \newlength{\gptboxwidth}%
      \newsavebox{\gptboxtext}%
    \fi%
    \setlength{\fboxrule}{0.5pt}%
    \setlength{\fboxsep}{1pt}%
    \definecolor{tbcol}{rgb}{1,1,1}%
\begin{picture}(5102.00,4534.00)%
    \gplgaddtomacro\gplbacktext{%
      \csname LTb\endcsname
      \put(378,936){\makebox(0,0)[r]{\strut{}\footnotesize 0.01}}%
      \put(378,1903){\makebox(0,0)[r]{\strut{}\footnotesize 0.03}}%
      \put(378,2870){\makebox(0,0)[r]{\strut{}\footnotesize 0.05}}%
      \put(378,3837){\makebox(0,0)[r]{\strut{}\footnotesize 0.07}}%
      \put(1054,310){\makebox(0,0){\strut{}\footnotesize 0.01}}%
      \put(2142,310){\makebox(0,0){\strut{}\footnotesize 0.03}}%
      \put(3230,310){\makebox(0,0){\strut{}\footnotesize 0.05}}%
      \put(4318,310){\makebox(0,0){\strut{}\footnotesize 0.07}}%
    }%
    \gplgaddtomacro\gplfronttext{%
      \csname LTb\endcsname
      \put(-359,2266){\rotatebox{-270}{\makebox(0,0){\strut{}\footnotesize SGV}}}%
      \put(2550,-20){\makebox(0,0){\strut{}\footnotesize EM}}%
    }%
    \gplbacktext
    \put(0,0){\includegraphics[width={255.10bp},height={226.70bp}]{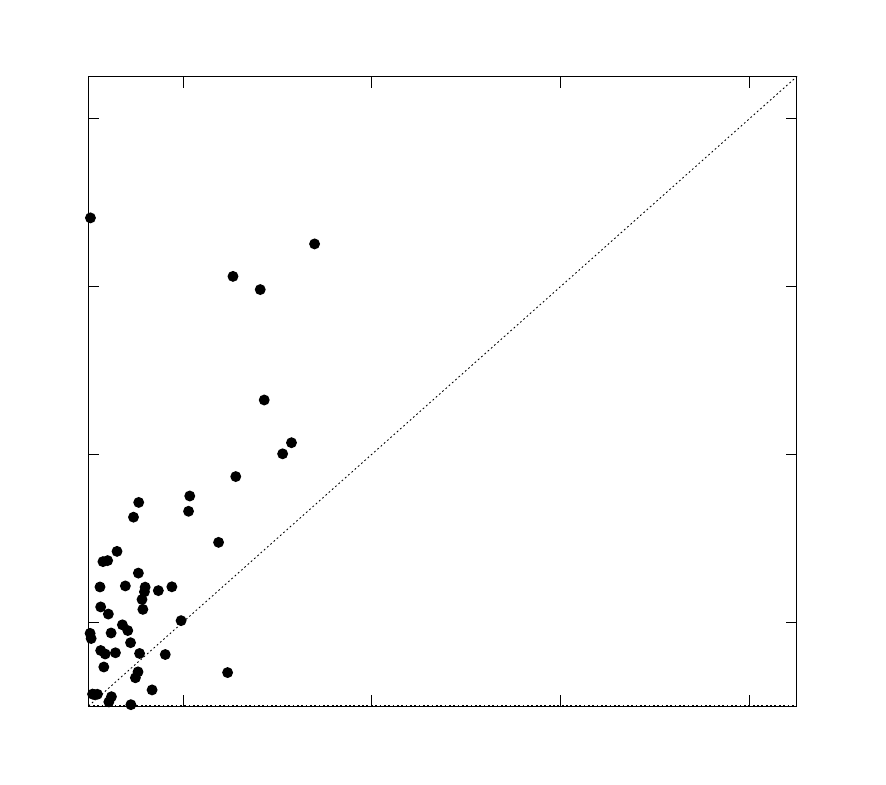}}%
    \gplfronttext
  \end{picture}%
\endgroup
  \caption{A comparison of the prediction performance for the value $z_c$
  located at the center point of the domain, $(1/2, 1/2)$. Each point has
  $(x,y)$ coordinates given by $\abs{\hat{z}_c(\bth_{\text{true}}) -
  \hat{z}_c(\hat{\bth})}$, where $\hat{\bth}$ denotes the estimated parameters
  using the EM method ($x$ coordinate) and the SGV method ($y$ coordinate).}
  \label{fig:centerinterp}
\end{figure}

And while assessing the significance of these improvements is not trivial, for a
very rough sense of their size, the prediction error variance of the process at
$\bx_c$ under the true law is approximately $0.0423$ without the nugget. The
mean absolute error between the conditional means was $0.0156$ for the SGV-based
predictions and $0.0062$ for the EM-based predictions, and so we see that their
difference of $0.0095$ is nearly a quarter the size of the prediction error
variance, which is not trivial.

\subsection{\rev{An Application to Doppler LIDAR data}}

\rev{
We now briefly show an application of this method to a real dataset using both a
nonstationary covariance model for the process $Z$ as well as for the
measurement error process $\eps$, so that $\bR$ is now a non-constant diagonal
matrix. The US Department of Energy's Atmospheric Radiation Measurement (ARM)
program provides high-resolution and high-frequency meteorological measurements
of a wide variety of quantities and at several locations throughout the US
\citep{stokes1994}. The Southern Great Plains (SGP) observatory is the largest
site, and among many other quantities it offers highly resolved vertical wind
profiles using Doppler LIDAR technology \citep{newsom2012,muradyan2020}. For a
more complete introduction and an example of a statistical study of this data,
see \citep{geoga2021}.  Figure \ref{fig:lidar} gives a visualization of the
data, highlighting in particular that the process' scale and measurement error
both clearly change with altitude.
}
\begin{figure}[!ht]
  \centering
\begingroup
  \makeatletter
  \providecommand\color[2][]{%
    \GenericError{(gnuplot) \space\space\space\@spaces}{%
      Package color not loaded in conjunction with
      terminal option `colourtext'%
    }{See the gnuplot documentation for explanation.%
    }{Either use 'blacktext' in gnuplot or load the package
      color.sty in LaTeX.}%
    \renewcommand\color[2][]{}%
  }%
  \providecommand\includegraphics[2][]{%
    \GenericError{(gnuplot) \space\space\space\@spaces}{%
      Package graphicx or graphics not loaded%
    }{See the gnuplot documentation for explanation.%
    }{The gnuplot epslatex terminal needs graphicx.sty or graphics.sty.}%
    \renewcommand\includegraphics[2][]{}%
  }%
  \providecommand\rotatebox[2]{#2}%
  \@ifundefined{ifGPcolor}{%
    \newif\ifGPcolor
    \GPcolortrue
  }{}%
  \@ifundefined{ifGPblacktext}{%
    \newif\ifGPblacktext
    \GPblacktexttrue
  }{}%
  \let\gplgaddtomacro\g@addto@macro
  \gdef\gplbacktext{}%
  \gdef\gplfronttext{}%
  \makeatother
  \ifGPblacktext
    \def\colorrgb#1{}%
    \def\colorgray#1{}%
  \else
    \ifGPcolor
      \def\colorrgb#1{\color[rgb]{#1}}%
      \def\colorgray#1{\color[gray]{#1}}%
      \expandafter\def\csname LTw\endcsname{\color{white}}%
      \expandafter\def\csname LTb\endcsname{\color{black}}%
      \expandafter\def\csname LTa\endcsname{\color{black}}%
      \expandafter\def\csname LT0\endcsname{\color[rgb]{1,0,0}}%
      \expandafter\def\csname LT1\endcsname{\color[rgb]{0,1,0}}%
      \expandafter\def\csname LT2\endcsname{\color[rgb]{0,0,1}}%
      \expandafter\def\csname LT3\endcsname{\color[rgb]{1,0,1}}%
      \expandafter\def\csname LT4\endcsname{\color[rgb]{0,1,1}}%
      \expandafter\def\csname LT5\endcsname{\color[rgb]{1,1,0}}%
      \expandafter\def\csname LT6\endcsname{\color[rgb]{0,0,0}}%
      \expandafter\def\csname LT7\endcsname{\color[rgb]{1,0.3,0}}%
      \expandafter\def\csname LT8\endcsname{\color[rgb]{0.5,0.5,0.5}}%
    \else
      \def\colorrgb#1{\color{black}}%
      \def\colorgray#1{\color[gray]{#1}}%
      \expandafter\def\csname LTw\endcsname{\color{white}}%
      \expandafter\def\csname LTb\endcsname{\color{black}}%
      \expandafter\def\csname LTa\endcsname{\color{black}}%
      \expandafter\def\csname LT0\endcsname{\color{black}}%
      \expandafter\def\csname LT1\endcsname{\color{black}}%
      \expandafter\def\csname LT2\endcsname{\color{black}}%
      \expandafter\def\csname LT3\endcsname{\color{black}}%
      \expandafter\def\csname LT4\endcsname{\color{black}}%
      \expandafter\def\csname LT5\endcsname{\color{black}}%
      \expandafter\def\csname LT6\endcsname{\color{black}}%
      \expandafter\def\csname LT7\endcsname{\color{black}}%
      \expandafter\def\csname LT8\endcsname{\color{black}}%
    \fi
  \fi
    \setlength{\unitlength}{0.0500bp}%
    \ifx\gptboxheight\undefined%
      \newlength{\gptboxheight}%
      \newlength{\gptboxwidth}%
      \newsavebox{\gptboxtext}%
    \fi%
    \setlength{\fboxrule}{0.5pt}%
    \setlength{\fboxsep}{1pt}%
    \definecolor{tbcol}{rgb}{1,1,1}%
\begin{picture}(8502.00,2266.00)%
    \gplgaddtomacro\gplbacktext{%
      \csname LTb\endcsname
      \put(359,587){\makebox(0,0)[r]{\strut{}\footnotesize 0.2}}%
      \put(359,870){\makebox(0,0)[r]{\strut{}\footnotesize 0.4}}%
      \put(359,1153){\makebox(0,0)[r]{\strut{}\footnotesize 0.6}}%
      \put(359,1436){\makebox(0,0)[r]{\strut{}\footnotesize 0.8}}%
      \put(359,1719){\makebox(0,0)[r]{\strut{}\footnotesize 1}}%
      \put(359,2002){\makebox(0,0)[r]{\strut{}\footnotesize 1.2}}%
      \put(1411,343){\makebox(0,0){\strut{}\footnotesize 14.85}}%
      \put(2875,343){\makebox(0,0){\strut{}\footnotesize 14.9}}%
      \put(4340,343){\makebox(0,0){\strut{}\footnotesize 14.95}}%
      \put(5805,343){\makebox(0,0){\strut{}\footnotesize 15}}%
      \put(7269,343){\makebox(0,0){\strut{}\footnotesize 15.05}}%
    }%
    \gplgaddtomacro\gplfronttext{%
      \csname LTb\endcsname
      \put(-154,1302){\rotatebox{-270}{\makebox(0,0){\strut{}\footnotesize Altitude (km)}}}%
      \put(3995,13){\makebox(0,0){\strut{}\footnotesize UTC Time (h)}}%
      \csname LTb\endcsname
      \put(8167,453){\makebox(0,0)[l]{\strut{}\footnotesize -4}}%
      \put(8167,665){\makebox(0,0)[l]{\strut{}\footnotesize -3}}%
      \put(8167,877){\makebox(0,0)[l]{\strut{}\footnotesize -2}}%
      \put(8167,1089){\makebox(0,0)[l]{\strut{}\footnotesize -1}}%
      \put(8167,1302){\makebox(0,0)[l]{\strut{}\footnotesize 0}}%
      \put(8167,1514){\makebox(0,0)[l]{\strut{}\footnotesize 1}}%
      \put(8167,1726){\makebox(0,0)[l]{\strut{}\footnotesize 2}}%
      \put(8167,1938){\makebox(0,0)[l]{\strut{}\footnotesize 3}}%
      \put(8167,2151){\makebox(0,0)[l]{\strut{}\footnotesize 4}}%
      \put(8562,1302){\rotatebox{-270}{\makebox(0,0){\strut{}\footnotesize Velocity (m/s)}}}%
    }%
    \gplbacktext
    \put(0,0){\includegraphics[width={425.10bp},height={113.30bp}]{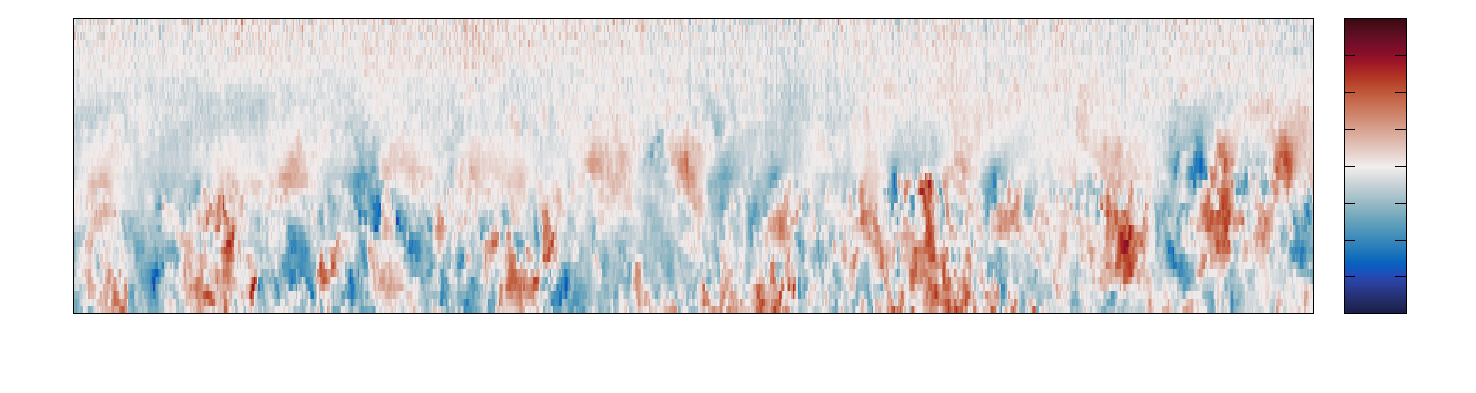}}%
    \gplfronttext
  \end{picture}%
\endgroup
  \caption{\rev{Doppler LIDAR vertical wind profiles from the main ARM field site
  (C1) in Oklahoma.}}
  \label{fig:lidar}
\end{figure}

\rev{
The model we will fit here, unlike in \cite{geoga2021}, is not concerned with
estimating the height of the atmospheric boundary layer (ABL). We instead use a
simple model for the space-time process given by
\begin{equation} \label{eq:lidar_model}
  Y(t, x) = \sigma(x) Z(t, x) + \eta(x) \eps(t, x),
\end{equation}
where $Z$ is a standard Mat\'ern model with a full geometric anisotropy
parameterized in terms of its inverse Cholesky factor, so that the geometric
anisotropy matrix $\bm{\Gamma}^{-1} = \bm{W} \bm{W}^T$ is represented by
parameters $W_{1,1} > 0$, $W_{1,2}$, and $W_{2,2} > 0$, and $\eps(t,x)$ is a
standard normal i.i.d.  noise whose standard deviation is modulated by the
spatial dependence of $\eta(x)$. In both cases, spatial dependence is modeled
simply with $\sigma(x) = \sum_{j=1}^3 w_j(x) \sigma_j$ (and similarly for
$\eta(x)$), where $w_j$ are normalized weights based on the distance from knot
points, placed at $0.2, 0.8$, and $1.2$ km. In total, the model has ten
parameters.
}

\rev{
The results of the estimation on the approximately $32{,}000$ measurements using
(i) independent blocks of size $80$ (two full vertical profiles), a natural
first method to try for trivial scalability and estimates from unbiased
estimating equations, (ii) a na\"ive Vecchia approximation that simply includes
the nugget, and (iii) the EM refinement of (ii), both computed with chunks of
size $20$ and three past chunks for conditioning, are shown in Table
\ref{tab:lidar_ests}. In all cases, the data are ordered as a vector time series
broken into two pieces (approximately, measurements below and then above
$0.6$m), and conditioning sets for each chunk were chosen from the most recent
past. As can be seen, the true likelihood of the data under this model is
nontrivially improved by the EM refinement, which considering that it retains
the theoretical guarantees with regard to bias thus gives an objectively better
estimator. Interestingly, while the block diagonal estimator is of course faster
and would still be faster even for much larger blocks, for this particular
dataset estimates with larger blocks happened to have lower terminal exact
likelihoods.
}

\begin{table}
  \centering
  \begin{tabular}{ccccccccccccc}
Model & $\ell_{\text{exact}}(\hat{\bth})$ & $\sigma_1$ & $\sigma_2$ & $\sigma_3$ & $W_{1,1}$ & $W_{1,2}$ & $W_{2,2}$ & $\nu$ & $\eta_1$ & $\eta_2$ & $\eta_3$ \\
\hline
Indept. Blocks  & -27317 &	 0.92  &	 0.36 &	 0.54 &	  8.30 &	0.17 & 773.36 & 0.61 &	 0.001 &	 0.01 &	 0.10 \\
Na\"ive Vecchia & -31898 &	 1.20  &	 0.48 &	 0.03 &	 12.26 &	 0.71 & 1069.43 &	 0.99 &	 0.001 &	 0.0003 &	 0.05 \\
EM + Vecchia    & \textbf{-31916} &	 1.12 &	 0.47 &	 0.04  &  14.41 &	 0.64 & 1272.12 &	 1.06 &	0.006 & $2 \times 10^{-5}$ &	 0.05 \\
\hline
  \end{tabular}
  \label{tab:lidar_ests}
  \caption{Exact negative log-likelihoods and point estimates for parameters of
  the model (\ref{eq:lidar_model}) for the Doppler LIDAR data shown in Figure
  \ref{fig:lidar} for three different estimation procedures: a block-diagonal
  covariance approximation (top row), na\"ive Vecchia estimates that use the
  error $\eps$ (middle row), and the EM-based estimator (bottom row).}
\end{table}

\rev{
As a final comment on this application, we remind the reader that this
covariance model, which still is in some ways inflexible, required ten
parameters.  A practitioner equipped with an $\bO(n)$ likelihood approximation
may still struggle to fit a model like this one if they are forced to use
gradient-free optimizers like Nelder-Mead, especially considering that many
natural parameterizations for covariance models give individual parameters
complex dependence relationships with others, making likelihood surfaces
particularly challenging to work with.  The method of this paper is much more
amenable to derivative-based optimization by virtue of providing an E function
that is easily automatically differentiated, and the accompanying software is
unique among its alternatives in that a practitioner can trivially provide their
own covariance functions---written in the same language as the rest of the
software, no less---and immediately obtain estimates computed with gradients and
Hessians.  The value of this flexibility and composability, and the freedom it
offers, is considerable.  
}

\section{Discussion} \label{sec:discussion}

We have introduced in this work a new method for performing parameter estimation
of Gaussian Process models whose likelihoods have been approximated in the
precision space but contain some kind of polluting noise which, if not handled
with special attention, severely reduces the accuracy of the precision-space
approximation \rev{or introduces problematic bias in the estimates}. This work
focused specifically on Vecchia's approximation \citep{vecchia1988} and additive
measurement error, but it applies with equal directness to any method that
provides sparse approximations for precision matrices and any polluting noise
whose covariance matrix at least admits a fast matrix-vector product.  This is a
broad category, including Markov random fields \citep{rue2005}, nearest neighbor
Gaussian processes \citep{datta2016,finley2019}, finite-element based SPDE
methods \citep{lindgren2011}, and surely many others. 

The method we propose here is potentially advantageous over its alternatives in
several ways. For one, \rev{its estimators provably correspond to the solution
of unbiased estimating equations, unlike any of its Vecchia- and
likelihood-based alternative besides simply ignoring that the nugget ruins
screening and using standard Vecchia approximations anyway}. Secondly, at least
in the case of Vecchia approximations, it removes all (numerical) sparse matrix
factorizations from the M steps in which optimization that benefits from
derivative information is performed. While not obviously impossible, a careful
implementation of the gradient and Hessian of the likelihood (or expected
likelihood) that is written in terms of Cholesky factors for the precision would
be very difficult to put together in a way that preserves performance,
parallelizability, and computational complexity. Practically speaking, the E
function given here that is optimized in each M step is the first objective
function in the setting of perturbed Vecchia approximations for which second
derivatives are conveniently computable, either by hand or by automatic
differentiation. In fact, to our knowledge this is the first publication that
uses true Hessian matrices of any Vecchia approximation, with the closest prior
work using expected Fisher information matrices \citep{guinness2021}. Moreover,
it makes use of the newly available automatic derivatives for $\mathcal{K}_\nu$
provided in \citep{geoga2022}, which can be composed at no extra effort with the
automatic derivatives of the E function itself and allow derivative based
estimation of smoothness parameters. 

\rev{
There are many questions that this work does not answer. Primarily, this method
can only be as good as the original approximation $\tbOm(\bth) \approx
\bS(\bth)^{-1}$, but particularly for non-adaptive approximation schemes like
Vecchia the quality of this approximation can vary.  For one example, separable
covariance functions that are popular in fields like computer experiments can
give rise to processes that do not screen well \citep{stein2011}. For kernels
that are smooth away from the origin, a property that is closely connected to
screening, many of the algebraic approximation tools discussed in the
introduction exploit the rank deficiency of off-diagonal matrix blocks. But that
level of rank deficiency changes with the dimension of the process (see
\cite{ambikasaran2016} for an example), making those algorithms less performant.
We expect this reduced approximation accuracy to carry over to precision-based
methods like Vecchia approximations. Another natural question to ask is how one
obtains standard errors for the resulting estimators. There is a straightforward
relationship between the Hessian of the log-likelihood for the data and
derivatives of the E and M functions \citep{dempster1977}, but considering that
we are using an approximated likelihood, this Hessian is almost certainly not
even asymptotically the precision of the MLE. Considering this additional
complication, we leave a more careful investigation of uncertainty
quantification in this setting to a future project.  With regard to efficiency,
in the symmetrized case using exact covariance matrices \cite{stein2013} gives a
bound on the efficiency cost of $1+S^{-1}$, which is quite small, and the
supplemental information shows similarly optimistic if not better results in the
setting of this work. The efficiency cost of Vecchia approximations themselves,
however, is an open area of study with only a few special cases having any
theory at all, and so this work cannot make any claims about the efficiency of
resulting estimators. Finally, we note that in the case of restricted maximum
likelihood estimation (REML), it is possible for the supremum of the likelihood
function to occur as the range tends to infinity \citep{stein2022}, which may
pose problems for iterative methods like the EM algorithm.
}

\rev{
Lastly, we observe that there are many potential improvements to our
approach here that are not discussed in this work. For example, many more
accurate methods for trace approximation exist outside of the Hutchinson
paradigm, like the peeling method \citep{lin2011}, first applied to GPs in
\citep{minden2016}. Adaptive Hutchinson- or Krylov-type methods also exist
\citep{meyer2021,persson2022,chen2022}, which may offer improvements in some
settings without requiring so much additional machinery.  Moreover, many
extensions, improvements, and generalizations of the EM algorithm exist (see
\citep{liu1994} for one of many examples), and it is conceivable that at least
some of those extensions may provide significant benefit in this setting. 
}

\renewcommand{\abstractname}{Acknowledgements}
\begin{abstract}
\noindent The authors thank Lydia Zoells for her careful copyediting.
\end{abstract}

\bibliography{references.bib}

\appendix

\section{Proofs} \label{sec:proofs}

\begin{proof}[Proof of Proposition $1$]
We first compute the conditional distribution of $\bz \sv \by$ under the
covariance function $K_{\bth_0}$, which is given in a particularly convenient
form as 
\begin{equation} \label{eq:zhat}
  \bz \sv \by,\bth_0 \sim \Nd\set{
    (\bS(\bth_0)^{-1} + \bR(\bth_0)^{-1})^{-1} \bR(\bth_0)^{-1} \by
    ,
    \;
    (\bS(\bth_0)^{-1} + \bR(\bth_0)^{-1})^{-1}
  }.
\end{equation}
To see this, we first observe that the joint distribution of $\by$ and $\bz$ is
given by
\begin{equation*} 
  \mat{
    \by \\ \bz
  }
  \sim
  \Nd\set{
    \mat{
      \bm{0} \\ \bm{0}
    },
    \;
    \mat{
      \bS + \bR & \bS \\
      \bS & \bS
    }
  },
\end{equation*}
where all matrices are assembled with $\bth_0$ and parameter indices will be
omitted when obvious or irrelevant.  The standard Gaussian conditional
distribution expression and the Sherman-Morrison-Woodbury lemma then gives that
$
  \bS - \bS (\bS + \bR)^{-1} \bS = (\bS^{-1} + \bR^{-1})^{-1}.
$
An even more basic identity of $\bS + \bR = \bR(\bS^{-1} + \bR^{-1})\bS$ gives
the conditional mean after multiplying through using the standard definition.

Next, we observe that 
\begin{equation*} 
  \mat{
    \bS + \bR & \bS \\
    \bS & \bS
  }
  =
  \mat{
    \I & \I \\
    & \I
  }
  \mat{
  \bR & \\
  & \bS
  }
  \mat{
    \I  \\
    \I & \I
  },
\end{equation*}
from which we immediately see that the log-determinant of the joint covariance
matrix is just $\log\abs{\bS} + \log\abs{\bR}$. From the further observation
that 
\begin{equation*} 
  \mat{
    \bS + \bR & \bS \\
    \bS & \bS
  }^{-1}
  =
  \mat{
    \I &  \\
    -\I & \I
  }
  \mat{
  \bR^{-1} & \\
  & \bS^{-1}
  }
  \mat{
    \I & -\I \\
    & \I
  },
\end{equation*}
we then can expand the quadratic form to be
\begin{equation*} 
  \mat{
    \by \\ \bz
  }^T
  \mat{
    \bS + \bR & \bS \\
    \bS & \bS
  }^{-1}
  \mat{
    \by \\ \bz
  }
  =
  (\by-\bz)^T \bR^{-1} (\by-\bz) + \bz^T \bS^{-1} \bz.
\end{equation*}
If $\bz$ has the distribution above given above, then the conditional
distribution of $\by-\bz$ given $\by$ is itself Gaussian with just a shifted
mean $\by-\hbz(\bth_0)$. The expectation of these two quadratic forms can then
be directly computed. For the sake of notational clarity, let $\bS_0$ and
$\bR_0$ denote matrices assembled with $\bth_0$. We then see that
\begin{align*} 
  \E_{\bz \sv \by, \bth_0} (\by-\bz)^T \bR^{-1} (\by-\bz)
  &= (\by - \hbz(\bth_0))^T \bR^{-1} (\by - \hbz(\bth_0)) 
  + \mathrm{tr}\left[ \bR^{-1} (\bS_0^{-1} + \bR_0^{-1})^{-1} \right], \quad
  \text{and}
  \\
  \E_{\bz \sv \by, \bth_0} \bz^T \bS^{-1} \bz
  &= \hbz(\bth_0)^T \bS^{-1} \hbz(\bth_0)
  + \mathrm{tr}\left( \bS^{-1} (\bS_0^{-1} + \bR_0^{-1})^{-1} \right).
\end{align*}
Combining these terms, including the trace terms, we reach the proposed
conclusion.
\end{proof}

\rev{
We now work towards the proof of Theorem $1$ with a sequence of lemmas. Because
the arguments are very general, they are given at no additional complexity in
much broader generality than is required for the desired result.
}
\begin{lemma}
\rev{
  Let $Z$ be a random process with a covariance function indexed by parameters
  $\bth$, and let $\gee(\bz)$ be a function with the following properties:
  \begin{itemize}
  \item[(i)] $\gee$ is integrable,
  \item[(ii)] $\gee$ is differentiable with respect to $\bth$,
  \item[(iii)] Each component of $\nabla_{\bth} \gee$ is bounded by an
  integrable function,  and
  \item[(iv)] $\nabla_{\bth} \gee$ yields unbiased estimating equations (UEEs)
  for $\bth$ under the law of $Z$ for all finite dimensional samples of $Z$ (so
  that, in particular, $\E_{\bth} \nabla_{\bth} \gee(\bz) = \bm{0}$, where $\bz$
  is any finite-dimensional sample of the process $Z$ and $\E_{\bth}$ is the
  expectation under the true law with parameters $\bth$).
  \end{itemize}
  Second, let $h_{\bm{\phi}}$ be any valid family of laws for a process $W \perp
  Z$ whose index $\bm{\phi}$ is disjoint from $\bth$. Then $\nabla_{\bth} [ \gee
  \star \ h_{\bm{\phi}}(\by) ]$ gives UEEs for all finite-dimensional samples of
  $Y = Z+W$ under the law of $Y = Z + W$.
}
\end{lemma}

\begin{proof}
\rev{
Picking an arbitrary finite-dimensional sample and starting from a simple
rewrite of the direct expression
\begin{equation*} 
  \E_{\bth} \nabla_{\bth} \int \gee(Y - \btau) h_{\bm{\phi}}(\btau) \dif \btau,
\end{equation*}
a sequence of simple manipulations gives that this expression equals
\begin{align*} 
  \E_{\bth} \nabla_{\bth} \E_{\bm{\phi}} \gee(Y - W)
  =
  \E_{\bm{\phi}} \E_{\bth} \nabla_{\bth} \gee(Y - W)
  =
  \E_{\bm{\phi}} \E_{\bth} E\left[\nabla_{\bth} \gee(Y - W) \sv W \right],
\end{align*}
where the first equality uses Fubini-Tonelli and dominated convergence (possible
due to requirement (iii) on $\gee$). But by the distributional assumption of the
theorem, $Y - W = Z$, and so this last term is equal to
\begin{equation*} 
  \E_{\bm{\phi}} \E_{\bth} 
  E\left[
  \nabla_{\bth} \gee(Z)
  \sv
  W
  \right]
  =
  \E_{\bm{\phi}} \E_{\bth} 
  \nabla_{\bth} \gee(Z),
\end{equation*}
where the conditional expectation can be dropped by the hypothesis that $Z \perp
W$.  Since the inner expectation is equal to zero by condition (ii) on $\gee$,
this quantity is zero and the proof is complete.
}
\end{proof}

\begin{lemma}
\rev{
Let $Z$ be a process whose finite-dimensional covariance matrices are indexed by
$\bS(\bth)$, and let $\tilde{\bS}(\bth)$ be a matrix-valued function yielding
full-rank matrices that are differentiable with respect to $\bth$ such that
\begin{equation} \label{eq:lemma2_uee}
  \E_{\bth} \left( \mathrm{tr} \set{ \tilde{\bS}(\bth)^{-1} \left[ \partial_{\theta_j} \tilde{\bS}(\bth)
  \right] } - 
  \bz^T  
  \tilde{\bS}(\bth)^{-1} \left[\partial_{\theta_j} \tilde{\bS}(\bth) \right] \tilde{\bS}(\bth)^{-1}
  \bz \right) = \bm{0},
\end{equation}
where $\E_{\bth}$ is an expectation under the true law of $Z$.  Second, let $W$
be a Gaussian process that is independent of $Z$ that yields finite-dimensional
distributions with covariance matrices $\bR(\bm{\phi})$. Then the gradient of
the function
\begin{equation*} 
  \tilde{\ell}_{\bth}(\by) 
  =
  \log \abs{\tilde{\bS}(\bth) + \bR(\bm{\phi})} 
  + 
  \by^T \left(
    \tilde{\bS}(\bth) + \bR(\bm{\phi})
  \right)^{-1} \by
\end{equation*}
yields UEEs for $\bth$ under the law of $Y = Z+W$.
}
\end{lemma}
\begin{proof}
\rev{
  Let $\gee$ be the multivariate normal density using $\tilde{\bS}(\bth)$ in
  place of $\bS(\bth)$. We first check the conditions to apply Lemma $1$. By
  virtue of being a Gaussian density, $\gee$ is integrable, so (i) is met.
  Requirement (ii) is an explicit assumption of the Lemma and is thus also met.
  Component (iii) can be confirmed by observing that 
  $
    \partial_{\theta_j} \gee(\bz)
    =
    \gee(\bz) \set{ \partial_{\theta_j} \log \gee(\bz)}.
  $
  With some expansion, $\int \partial_{\theta_j} \gee$ is proportional to
  \begin{equation*} 
    A(\bth) + C(\bth) \int e^{-\bz^T \tilde{\bS}(\bth)^{-1} \bz}
    \cdot
    \bz^T 
    \tilde{\bS}(\bth)^{-1}
    [\partial_{\theta_j} \tilde{\bS}(\bth)]
    \tilde{\bS}(\bth)^{-1}
    \bz
    \; \;
    \dif \bz,
  \end{equation*}
  where $A(\bth)$ and $C(\bth)$ are constants that depend on $\bth$ but not
  $\bz$. From this form, we see that a sufficient condition for this integral to
  be finite is that $\tilde{\bS}(\bth)$ is full rank (an assumption) and that
  $\norm{\tilde{\bS}(\bth)^{-1} [\partial_{\theta_j} \tilde{\bS}(\bth)]
  \tilde{\bS}(\bth)^{-1}}$ is finite. But this latter norm condition is
  guaranteed by the differentiability of $\tilde{\bS}(\bth)$, and so (iii) is
  satisfied. Since (iv) is also an assumption of the Lemma, all conditions of
  Lemma $1$ have been satisfied.
  For the second step, let $h_{\bm{\phi}}$ be the density corresponding to $W$.
  With these choices of $\gee$ and $h_{\bm{\phi}}$ and the observation that $Z
  \perp W$ means that all finite-dimensional marginals of $Z$ and $W$ will be
  jointly Gaussian and that $\tilde{\ell}_{\bth}$ is nothing but $\log [\gee
  \star h_{\bm{\phi}}]$. And since $\gee \star h_{\bm{\phi}}$ is supported on
  the entire plane, $\nabla_{\bth} [\gee \star h_{\bm{\phi}}] = 0$ if and only
  if $\nabla_{\bth} \log [\gee \star h_{\bm{\phi}}]$ does. Applying Lemma $1$ to
  $\gee \star h_{\bm{\phi}}$ thus gives the desired conclusion.
}
\end{proof}

\rev{
Finally, the proof of the theorem reduces to the following:
\begin{proof}[Proof of Theorem $1$]
  Claim $1$ follows from the elementary property of the EM algorithm that each M
  step provably does not decrease the marginal log-likelihood (Theorem $1$ in
  \citep{dempster1977}). For claim $2$, we note that the hypotheses of the
  Theorem on $\tbOm(\bth) \approx \bS(\bth)^{-1}$ satisfy the requirements of
  Lemma $2$.  Since Theorem $1$ further assumes that a parameter is used for
  either $\bS$ or $\bR$ but never both, we simply partition $\bth$ into separate
  components for $Z$ and $\eps$, denoted $\bth'$ and $\bm{\phi}$, and apply
  Lemma $2$ twice. In the first application, we let $(Z, \bth')$ correspond to
  $(Z, \bth)$ in Lemma $2$ and $(\eps, \bm{\phi})$ to $(W, \bm{\phi})$. In the
  second application, we use the reverse association.
\end{proof}
}

\newpage

\section{Supplemental Material} \label{sec:supp}

\rev{
The following figure shows the path of EM iteration for $S =
\set{5,25,50,75,100,125}$ SAA vectors (where each larger set adds to the past
one, as opposed to completely re-drawing new ones) for the first $10$ simulated
datasets used in the simulation study against SGV.
}

\begin{figure}[!ht]
  \centering
  \includegraphics[width=0.9\textwidth]{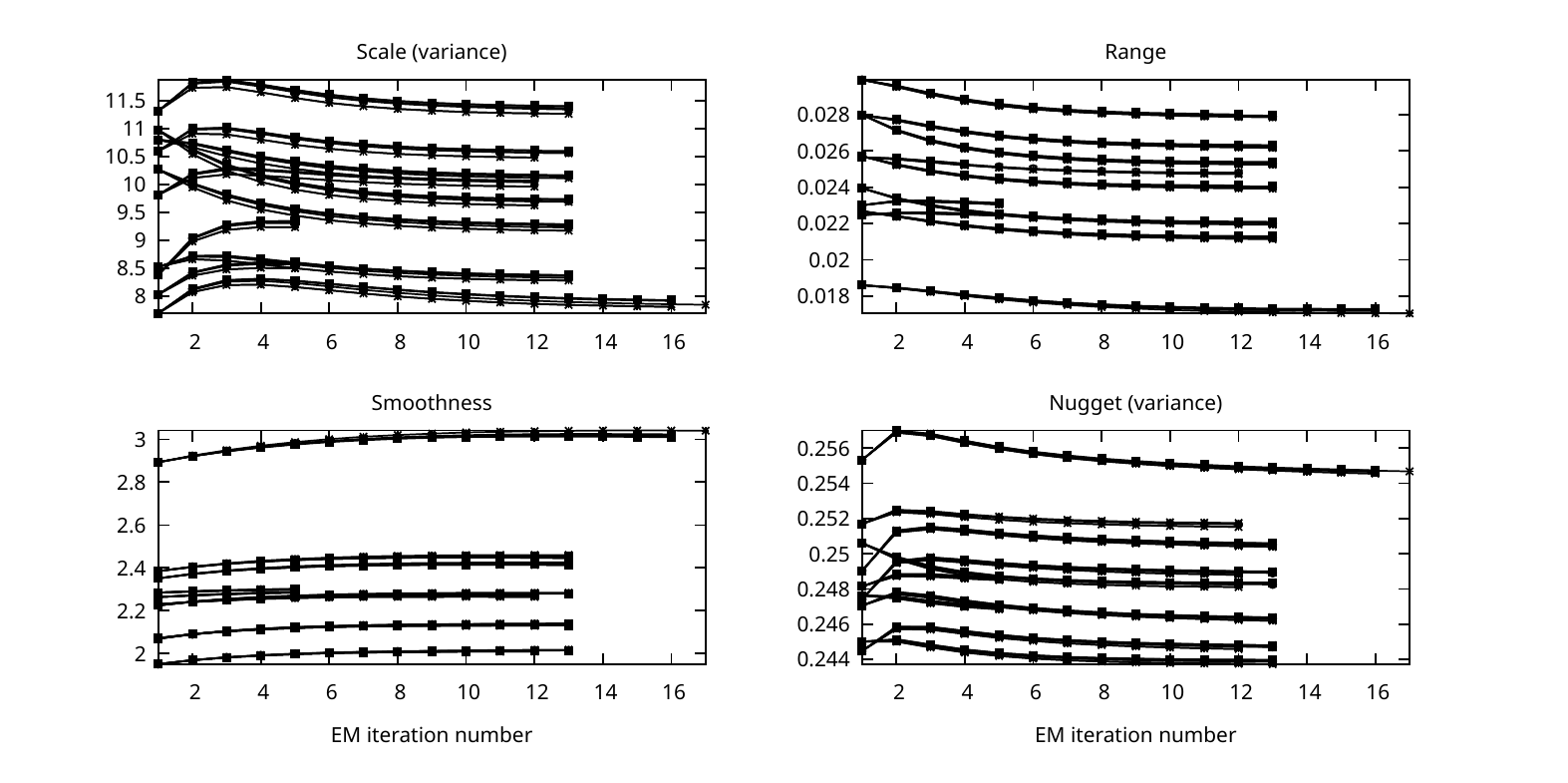}
\end{figure}

\rev{
This figure shows $10 \times 6 \times 4 = 240$ lines in total, but the important
observation to make is that the $6$ lines that correspond to the EM path for a
parameter using the same simulated data are very well-clustered, particularly
when compared to the natural variability of the estimate for the given data size
and model, which can be roughly inferred simply by studying the y-axis range of
the plots. While this is again not proof of anything, it is encouraging that the
number of SAA vectors seems very unlikely to levy a bothersome efficiency cost
on the resulting estimators.
The below figure shows a similar summary, now giving only a scatter plot of the
final estimate for each collection of SAA vectors (in increasing order) and each
data trial. The conclusion is again clear: the efficiency cost of the stochastic
trace, even when using only five sampling vectors, is negligible compared to the
intrinsic variability of the estimator for this model and data setting.
}

\begin{figure}[!ht]
  \centering
  \includegraphics[width=0.9\textwidth]{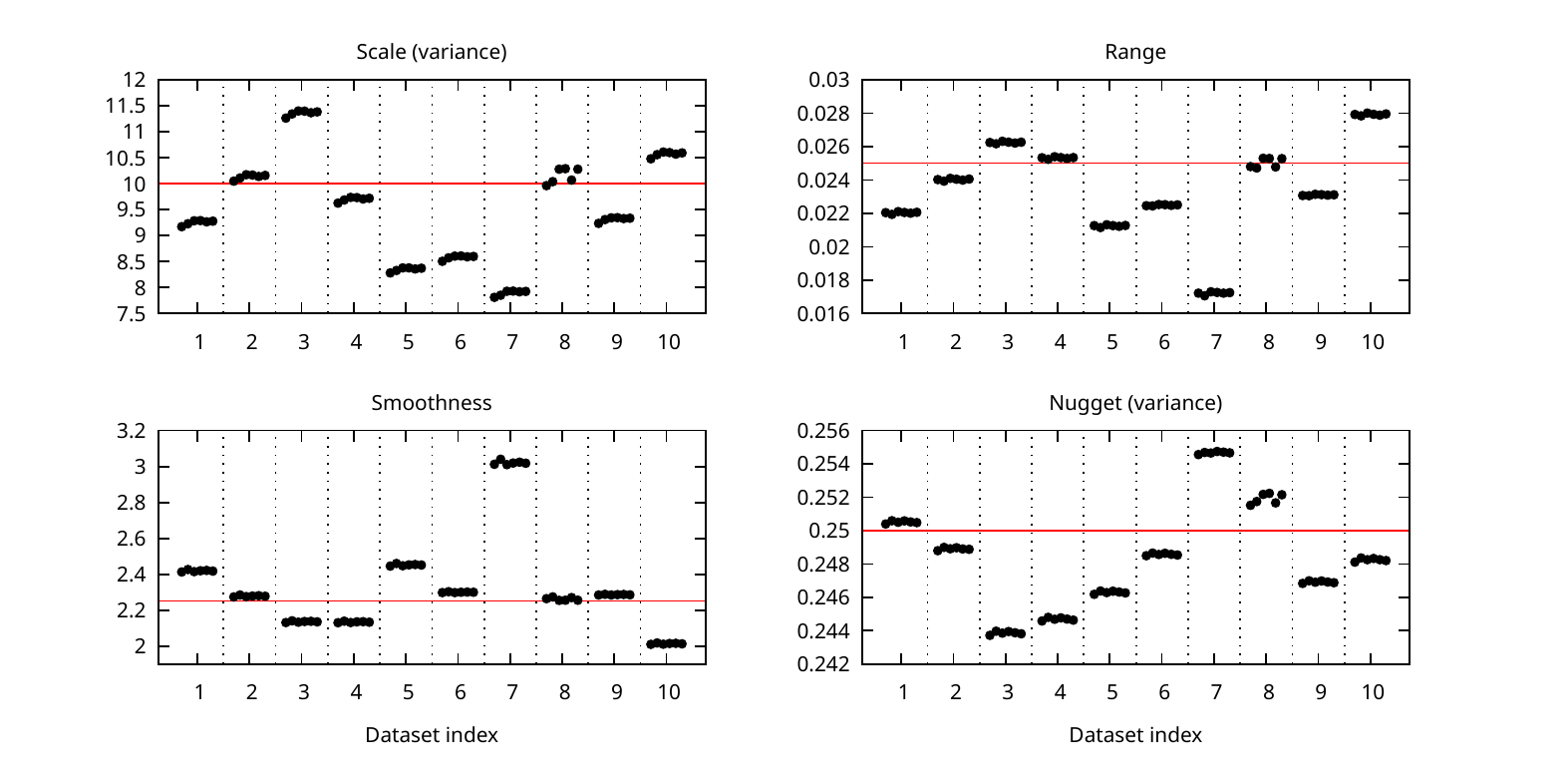}
\end{figure}

\end{document}